\documentclass[11pt]{article}
\usepackage{amsmath}
\usepackage{graphicx,psfrag,epsf}
\usepackage{enumerate}
\usepackage{natbib}
\usepackage{url} 

\newcommand{\blind}{0}

\usepackage{bm} 
\usepackage{scalerel}

\usepackage{xcolor}
 \usepackage{multirow}
 \usepackage{arydshln}
\usepackage{titlesec} 
\titlelabel{\thetitle.\quad} 
\titleformat*{\section}{\bf\Large\center}

\usepackage{eufrak}
\usepackage{float}
\usepackage{graphicx} 
\usepackage{bbm}
\usepackage{latexsym}
\usepackage{caption}
\usepackage[margin=20pt]{subcaption}
\usepackage{enumerate}

\graphicspath{ {./Graphs/} }
\usepackage{xr}

\newcommand{\mix}{\textup{Mix}}

\newcommand{\black}{\color{black}}
\newcommand{\hvl}{\hat v_\lin}
\newcommand{\hvn}{\hat v_\nm}

\newcommand{\hvs}{\hat v_*}
\newcommand{\pia}{\pi_a}
\newcommand{\nflt}{* = \textup{N}, \textup{F}, \textup{L}}

\newcommand{\trn}{\tilde\rho_\nm}

\newcommand{\tvl}{\tilde v_\lin}
\newcommand{\tvn}{\tilde v_\nm}
\newcommand{\tvf}{\tilde v_\fisher}
\newcommand{\ci}{\textup{CI}}
\newcommand{\qa}{q_{1-\alpha/2}}
\def\beginp{\begin{pmatrix}}
\def\endp{\end{pmatrix}}

\newcommand{\vs}{v_*}
 
\newcommand{\pfrt}{p_\textup{frt}}
\newcommand{\limN}{\lim_{N\to\infty}}

\newcommand{\frt}{\textsc{frt}}
\newcommand{\mz}{\mathcal Z}

\newcommand{\phts}{{preliminary-test }}
\newcommand{\limn}{\lim_{N\to\infty}}
\newcommand{\hf}{H_{0\fisher}}
\newcommand{\hn}{H_{0\nm}}
\newcommand{\HF}{H_{0\fisher}}
\newcommand{\HN}{H_{0\nm}}

\newcommand{\precreapp}{Assume a completely randomized treatment-control experiment and Condition \ref{cond:asym}}
\newcommand{\prefrt}{Assume a completely randomized treatment-control experiment and Conditions \ref{cond:asym}--\ref{cond:asym_frt}}
\newcommand{\cis}{\textup{CI}_{*}}
\newcommand{\cil}{\textup{CI}_{\lin}}
\newcommand{\cif}{\textup{CI}_{\fisher}}

\newcommand{\pt}{\textup{pt}}
\newcommand{\htptf}{\htau_{\pt, \fisher}}
\newcommand{\htptl}{\htau_{\pt,\lin}}
\newcommand{\htpts}{\htau_{\pt,*}}
\newcommand{\hsepts}{\hse_{\pt,*}}
\newcommand{\hseptl}{\hse_{\pt,\lin}}
\newcommand{\hseptf}{\hse_{\pt,\fisher}}

\newcommand{\cipts}{\textup{CI}_{\pt,*}}
\newcommand{\ciptl}{\textup{CI}_{\pt,\lin}}
\newcommand{\ciptf}{\textup{CI}_{\pt,\fisher}}

\newcommand{\hses}{\hse_*}
\newcommand{\mrm}{\mathbb R^m}
\newcommand{\meaniz}{N_z ^{-1}\sum_{i:Z_i = z}}

\newcommand{\cipt}{\textup{CI}_{\textup{pt}}}
\newcommand{\cin}{\textup{CI}_{\nm}}
\newcommand{\cia}{\textup{CI}_{\textup{adj}}}

\newcommand{\hsen}{\hse_\nm}
\newcommand{\hsef}{\hse_\fisher}
\newcommand{\hsel}{\hse_\lin}
\newcommand{\hsea}{\hse_{\textup{adj}}}
\newcommand{\hta}{\htau_{\textup{adj}}}

\newcommand{\hsept}{\hse_{\textup{pt}}}
\newcommand{\htpt}{\htau_{\textup{pt}}}

\def\hg{\hat\gamma}

\def\hvl{\hat v_\lin}

\def\mq{\mathcal Q}

\def\meani{N^{-1}\sumi}

\def\yiz{Y_i(z)}

\def\fl{* = \fisher, \lin}
\def\hgs{\hat\gamma_*}
\def\succi{\succeq_\infty}

\def\beginp{\begin{pmatrix}}
\def\endp{\end{pmatrix}}

\def\lmt{\texttt{lm}}

\def\nf{* = \neyman, \fisher}
\def\nfl{* = \neyman, \fisher, \lin}

\def\gp{\gamma_\fisher}

\newcommand{\mc}{\mathcal C}

\def\mi{\mathcal I}

\def\by{\bar Y}

\def\phim{\phi_{\textsc{m}}}

\def\mhls{Mahalanobis}

\def\ppinv{(\pp)^{-1}}

\def\vl{v_\lin}
\def\vf{v_\fisher}
\def\vn{v_\nm}

\def\cs{c_{*}}

\def\csi{c_{*}}
\def\cfi{c_{\fisher}}
\def\cni{c_{\nm}}
\def\cli{c_{\lin}}

\def\rs{\rightsquigarrow}

\def\pp{e_0e_1}

\def\sqrtn{\sqrt N}
\def\hts{\htau_*}
\def\begini{\begin{itemize}}
\def\endi{\end{itemize}}
\def\begine{\begin{enumerate}}
\def\ende{\end{enumerate}}

\def\ep{\epsilon}

\def\vxi{v_{x}}

\def\mi{\mathcal{I}}

\def\vx{v_x}

\def\hth{\hat\theta}

\def\hts{\htau_*}

\def\mn{\mathcal{N}}

\def\neyman{\textsc{n}}
\def\rtn{\sqrt N}

\def\hse{\hat{\textup{se}}}

\newcommand{\hgf}{\hat\gamma_\fisher}

\def\hgs{\hg_*}

\newcommand{\hglo}{\hat\gamma_{\lin,1}}
\newcommand{\hglz}{\hat\gamma_{\lin,0}}
\newcommand{\hglzz}{\hat\gamma_{\lin,z}}

\newcommand{\htf}{\hat\tau_\fisher}
\newcommand{\htl}{\hat\tau_\lin}
\def\hts{\hat\tau_*}

\newcommand{\htn}{\hat\tau_\nm}

\def\Zi{Z_i}
\def\Yi{Y_i}
\def\sumi{\sum_{i=1}^N}
\def\sumN{\sum_{i=1}^N}

\def\gq{\gamma_z}

\def\ml{\mathcal{L}}

\def\sm{{Supplementary Material}}

\def\hg{\hat\gamma}

\def\sxx{S^2_x}

\def\htx{\htau_x}

\def\ms{* = \nm, \fisher,\lin}

\newcommand{\htau}{ \hat\tau}
\def\htf{ \hat\tau_\fisher}

\def\hx{\hat x}

\def\mt{\mathcal{T}}

\newcommand{\mN}{\mathcal{N}}

\newcommand{\op}{o_{\pr}(1)}

\newcommand{\lin}{{\textsc{l}}}
\newcommand{\nm}{{\textsc{n}}}

\newcommand{\bg}{  \gamma }

\newcommand{\hY}{\hat Y}

\newcommand{\pr}{{\mathbb P}}
 \newcommand{\ot}[1]{1, \ldots,#1}

\def\T{{ \mathrm{\scriptscriptstyle T} }}

\newcommand{\cov}{\text{cov} }

\newcommand{\fisher}{{\textsc{f}}}

\newcommand{\asim}{\overset{\cdot}{\sim}}

\newcommand{\mL}{{\mathcal{L}}}

\def\begina{\begin{eqnarray*}}
\def\enda{\end{eqnarray*}}
\def\beginy{\begin{eqnarray}}
\def\endy{\end{eqnarray}}
\def\begine{\begin{enumerate}}
\def\ende{\end{enumerate}}

\usepackage[utf8]{inputenc}
\usepackage[english]{babel}

\usepackage{setspace}
\usepackage{amsthm}
\usepackage{amssymb}
\usepackage{color}

\usepackage{comment}
\theoremstyle{definition}

\newtheorem*{theorem*}{Theorem}
\newtheorem{theorem}{Theorem}
\newtheorem*{rmk*}{remark}
\newtheorem{proposition}{Proposition}
\newtheorem{lemma}{Lemma}

\newtheorem{condition}{Condition}

\newtheorem{definition}{Definition}

\newtheorem*{corollary*}{Corollary}

\usepackage{natbib} 
\bibpunct{(}{)}{;}{a}{}{,} 

\usepackage{etoolbox} 
\apptocmd{\sloppy}{\hbadness 10000\relax}{}{} 

\usepackage{multibib}
\newcites{sec}{References}

\usepackage{color}
\usepackage{listings}
\usepackage{hyperref}
\usepackage{booktabs}
\usepackage{lscape}

\RequirePackage[normalem]{ulem}

\allowdisplaybreaks

\begin{document}

 \def\spacingset#1{\renewcommand{\baselinestretch}%
{#1}\small\normalsize} \spacingset{1}


\if0\blind
{
  \title{\bf A randomization-based theory for preliminary testing of covariate balance in controlled trials}
  \author{Anqi Zhao\thanks{Department of Statistics and Data Science, National University of Singapore (NUS)} \ and 
Peng Ding\thanks{Department of Statistics, University of California, Berkeley. The authors gratefully acknowledge the Start-Up grant R-155-000-216-133 from NUS and the U.S. National Science Foundation (grant \# 1945136).  We thank the editor, the associate editor, and  three referees for their most constructive comments.}}
    \date{}
  \maketitle
} \fi

\if1\blind
{
  \bigskip
  \bigskip
  \bigskip
  \begin{center}
    {\LARGE\bf A randomization-based theory for preliminary testing of covariate balance in controlled trials}
\end{center}
  \medskip
} \fi

\bigskip
\begin{abstract}
Randomized trials balance all covariates on average and provide the gold standard for estimating treatment effects. Chance imbalances nevertheless exist more or less in realized treatment allocations and  intrigue an important question: what should we do in case the treatment groups differ with respect to some important baseline characteristics? A common strategy is to conduct a {\it preliminary test} of the balance of baseline covariates after randomization, and invoke covariate adjustment for subsequent inference if and only if the realized allocation fails some prespecified criterion. 
Although such practice is intuitive and popular among practitioners,  the existing literature has so far only evaluated its properties under strong parametric model assumptions in theory and simulation,  yielding results of limited generality. 
To fill this gap, we examine two strategies for conducting preliminary test-based covariate adjustment by regression, and evaluate the validity and efficiency of the resulting inferences from the randomization-based perspective. As it turns out, the preliminary-test estimator based on the analysis of covariance can be even less efficient than the unadjusted difference in means, and risks anticonservative confidence intervals based on normal approximation even with the robust standard error. 
The preliminary-test estimator based on the fully interacted specification is on the other hand less efficient than its counterpart under the {\it always-adjust} strategy, and yields overconservative confidence intervals based on normal approximation. 
In addition,  although the Fisher randomization test is still finite-sample exact for testing the sharp null hypothesis of no treatment effect on any individual,  it is no longer valid for testing the weak null hypothesis of zero average treatment effect in large samples even with properly studentized test statistics. These undesirable properties are due to the asymptotic non-normality of the preliminary-test estimators. Based on theory and simulation, we echo the existing literature and do not recommend the preliminary-test procedure for covariate adjustment in randomized trials. 
\end{abstract}

\noindent%
\vfill

\newpage
\spacingset{1.45} 

%
\section{Introduction}
%
\subsection{Preliminary test of covariate balance}
Randomized trials balance all observed and unobserved covariates on average, providing the gold standard for estimating treatment effects \citep{Fisher35, gerber2012field, rosenberger2015randomization}. 
Chance imbalances nevertheless exist more or less in realized treatment allocations \citep{altman, senn89, senn, morgan2012rerandomization},  complicating the interpretation of experimental results; see \citet[Section 7]{PostStratYu} for a canonical example from a clinical trial. 
What should we do in case the treatment groups differ with respect to some important baseline characteristics?

There are in general three common responses to this question. 
The {\it unadjusted} strategy 
trusts the stochastic balance ensured by the randomization mechanism, and uses the difference in average outcomes across treatment groups to estimate the average treatment effect without adjusting for covariates. 
The resulting inference is unbiased without any assumptions on the outcome-generating process \citep{Neyman23, CausalImbens, DingCLT}.   
The {\it always-adjust} strategy acknowledges the merits of covariate adjustment in improving the power and efficiency of inference, and always adjusts for all baseline covariates regardless of their balance across treatment groups \citep{Fisher35, freedman, Lin13, LD20}.
The  {\it preliminary-test} strategy takes a middle ground between the unadjusted and always-adjust strategies, and invokes covariate adjustment for inference if and only if the realized allocation fails some prespecified balance test \citep{posthoc, beach, altman2, permutt, mutz}.

Although the preliminary-test strategy is intuitive 
 and popular among practitioners, the theoretical literature holds a negative view on its implications on inference and discourages its use for experimental data in general. 
\cite{altman}, \cite{begg}, and \cite{senn} questioned the logical foundation of preliminary tests for randomized trials, and argued that such practice is illogical by testing a null hypothesis that is known to be true. 
In particular, with the randomization mechanism stochastically balancing all covariates between treatment groups on average, any significant result is by definition a false positive \citep{begg}.
The \textsc{consort} guidelines \citep{consort} echoed \cite{altman} and suggested that preliminary tests are superfluous and can be misleading for interpreting the experimental results. 
They accordingly recommended reporting tables of baseline demographic and clinical characteristics by experimental condition but discouraged testing for baseline covariate balance; see also \cite{boer}.  
\cite{senn89} examined the effect of covariate imbalance on the type I error rates of a test of treatment efficacy in randomized trials under a joint Gaussian model of outcome and a single covariate, and concluded that preliminary tests of covariate balance are not guaranteed to control the type I error rates. 
\cite{permutt} examined the operating characteristics of the \phts procedure under a similar model, and concluded that the \phts procedure has lower power than the always-adjust strategy.
\cite{mutz} provided an informative account of competing opinions in the medical and political science literature regarding the use of covariate balance test for experimental data, 
and concluded that one should not perform such tests unless the experimental data are compromised by complications like faulty randomization or differential attrition.
Despite the strong objections to preliminary tests in the clinical trial literature, however, most theoretical investigations so far were conducted under restricted parametric models \citep{senn89, permutt, mutz}, limiting the generality of subsequent findings. 
To fill this gap, we take a unified look at the role of covariate balance in controlled trials and clarify the implications of the preliminary-test procedure from the randomization-based perspective.

\black
\subsection{Our contributions}
This article makes several contributions. 
First, we extend the discussion of \cite{permutt} and \cite{mutz} to more general, non-parametric settings and formalize two schemes for preliminary test-based regression adjustment.
The proposed procedures use the {\mhls} distance of the difference in covariate means to form the covariate balance criterion, and invoke the additive \citep{Fisher35} and fully interacted \citep{Lin13} specifications for regression adjustment, respectively, in case the realized allocation fails the balance test.  

Second, we establish the sampling properties of the resulting point and interval estimators from the randomization-based perspective, and quantify their validity and efficiency relative to the unadjusted and always-adjust counterparts. 
The main findings are twofold. 
First, the preliminary-test estimator based on \cite{Fisher35}'s analysis of covariance can be even less efficient than the unadjusted difference in means and, moreover, risks anticonservative confidence intervals based on normal approximation even with the robust standard error.
Second, the preliminary-test estimator based on \cite{Lin13}'s fully interacted specification is less efficient than its counterpart under the always-adjust strategy and yields overconservative  confidence intervals based on normal approximation.

We then extend the discussion to the Fisher randomization test, and examine the validity of the preliminary-test estimators for testing not only the sharp null hypothesis of no treatment effect on any individual \citep{Fisher35, rubin80, CausalImbens} but also the weak null hypothesis of zero average treatment effect \citep{wd, zdfrt, colin}.  
As it turns out, the asymptotic non-normality of the preliminary-test estimators renders the usual strategy of studentization by Eicker--Huber--White heteroskedasticity-robust standard errors no longer sufficient to ensure the asymptotic validity of the tests for the weak null hypothesis. This illustrates the additional complications incurred by the preliminary-test procedure.  

Based on theory and simulation, we discourage the use of the preliminary-test procedure for covariate adjustment in randomized trials and recommend always adjusting for baseline covariates by \cite{Lin13}'s specification for both estimation and hypothesis testing.

\subsection{Notation and definitions}
We use the notation from \cite{zdrep} and \cite{zdca} to facilitate discussion. 
For a given set of tuples $\{(u_i, v_i): u_i \in \mathbb R, \ v_i \in \mathbb R^J; \ i = \ot{N}\}$, denote by $\lmt(u_i \sim v_i)$ the least squares  fit of the linear regression of $u_i$ on $v_i$ over $i = \ot{N}$. 
We focus on the numeric outputs of  least squares without any modeling assumption and evaluate their sampling properties from the randomization-based perspective.

Let $\qa $ denote the $(1-\alpha/2)$th quantile of the standard normal distribution. 
Let $\mi(\cdot)$ denote the indicator function. 
For two random variables $A$ and $B$ with cumulative distribution functions $F_A(t) = \pr(A\leq t)$ and $F_B(t) = \pr(B\leq t)$ for $t\in \mathbb R$, let 
\begina
\mix\beginp
A&: &\pi\\
B&: &1-\pi
\endp
\enda
denote the mixture of $A$ and $B$ with weights $(\pi, 1-\pi)$ for $\pi \in [0,1]$. The corresponding cumulative distribution function is $F(t) = \pi F_A(t) + (1-\pi) F_B(t)$. 

Assume that $\hth_1 \in \mathbb R$ and $\hth_2 \in \mathbb R$ are two consistent estimators for some parameter $\theta \in \mathbb R$. 
We say  $\hth_1$ is  {\it asymptotically more efficient} than $\hth_2$ if asymptotically, $\hth_1 - \theta$ has smaller central quantile ranges than  $\hth_2 - \theta$.
An asymptotically more efficient estimator has smaller asymptotic variance. 
For $\hth_1$ and $\hth_2$ that are both asymptotically normal, $\hth_1$ is asymptotically more efficient than $\hth_2$ if and only if the asymptotic variance of $\hth_1$ is smaller than that of $\hth_2$.

To avoid excessive notation, we state the theorems in general terms in the main paper and relegate the rigorous yet technical statements to the {\sm}. 
In particular, all theoretical results in the main paper assume a completely randomized treatment-control trial with a given set of regularity conditions for asymptotic analysis.
We give the definition of completely randomized treatment-control trials in Section \ref{sec:setup} and relegate the regularity conditions to the {\sm}.

\section{Basic setting of the treatment-control trial}\label{sec:setup}
We review in this section the basic setting of completely randomized treatment-control trials under the potential outcomes framework \citep{Neyman23}, and give an overview of the classical results on regression-based causal inference. 
We then formalize the preliminary-test procedure as the focus of our discussion.

\subsection{Potential outcomes and treatment effects}
Consider an intervention of two levels, $z = 0$ (control), $1$ (treatment), and a finite population of $N$ units, indexed by $i =1,\ldots, N$.  
We invoke the potential outcomes framework and denote by $\yiz  \in  \mathbb R$  the potential outcome of unit $i$ under treatment level $z \in \{0,1\}$. This notation assumes the stable unit treatment value assumption \citep{rubin80} that the potential outcomes for any unit do not vary with the treatments assigned to other
units, and, for each unit, there are no different forms or versions of each treatment level,
which lead to different potential outcomes. 
The individual treatment effect is $\tau_i  = Y_i(1) - Y_i(0)$, and the finite population average treatment effect is $\tau = N^{-1} \sumN  \tau_i$. 
Let $x_i = (x_{i1}, \ldots, x_{iJ})^\T  \in \mathbb R^J$ be the $J$-dimensional covariate vector for unit $i$. 
We center the $x_i$'s with $\bar x = N^{-1}\sumi x_i = 0_J$ to simplify the presentation.

Let $Z_i \in \{0,1\}$ denote the treatment level received by unit $i$.
For some prespecified, fixed treatment sizes $N_1 >0$ and $N_0 = N-N_1 >0$, complete randomization draws completely at random $N_1$ units to receive treatment and assigns the remaining $N_0$ units to control. 
The observed outcome is $ Y_i  = \Zi\Yi(1) + (1-\Zi)\Yi(0) $ for unit $i$.

Let $X = (x_1, \ldots, x_N)^\T  \in \mathbb R^{N\times J}$, $Z = (Z_1, \dots, Z_N)^\T \in \mathbb R^N$, and $Y = (Y_1, \ldots, Y_N)^\T \in \mathbb R^N$ summarize the covariates, treatment assignments, and observed outcomes of all $N$ units. 
Let $e_z = N_z /N$ for $z=0,1$. 

\subsection{Regression analysis and review of classical results}\label{sec:review}
Let $\hY(z) = \meaniz   Y_i $ denote the average observed outcome under treatment level $z$. 
The difference in means $\hat{\tau}_\nm  = \hY(1) - \hY(0)$ is unbiased for $\tau$ under complete randomization \citep{Neyman23}, and can be computed as the coefficient of $Z_i$ from the simple linear regression $\lmt(Y_i \sim 1 + Z_i)$. 

The presence of covariates promises the opportunity to improve estimation efficiency. 
\citet{Fisher35} suggested an estimator $\htau_\fisher$ for $\tau$, which equals  the coefficient of $Z_i$ from the additive linear regression  $\lmt(Y_i \sim 1 + Z_i + x_i)$. 
\citet{Lin13} recommended an alternative estimator, $\htau_\lin$, as the coefficient of $Z_i$ from the fully interacted linear regression  $\lmt(Y_i \sim   1 + Z_i + x_i + Z_i x_i)$. 
We use the subscripts \textsc{n}, \textsc{f}, and \textsc{l} to signify \cite{Neyman23}, \cite{Fisher35}, and \cite{Lin13}, respectively. 

Standard theory ensures that 
$\hts \ (\nfl)$ are all asymptotically normal with mean $\tau$ 
and the asymptotic variance of $\htl$ is the smallest among the three \citep{Neyman23, DingCLT, Freedman08a, Lin13}; we give the explicit forms of their asymptotic variances in the supplementary material. 
This ensures the asymptotic efficiency of $\htl$ over $\htn$ and $\htf$.
In addition, the corresponding Eicker--Huber--White heteroskedasticity-robust standard errors from least squares, denoted by $\hse_*$ for $\nfl$, are asymptotically  conservative for estimating the true sampling standard errors 
of the $\hts$'s.
This justifies the large-sample Wald-type inference based on $(\hts, \hses)$ and normal approximation, with the resulting $1-\alpha$ large-sample confidence intervals taking the form 
\begina
\cis = \left[\hts - \qa  \hses, \ \hts + \qa  \hses \right] \quad(\nfl).
\enda
Asymptotically, the limit of $\hsel/\hses$ is less than or equal to 1 for $\nf$ \citep[Corollary 3]{zdfrt} such that $\cil$ is the shortest among $\{\cis: \nfl\}$. Importantly, these theoretical guarantees are all randomization-based and hold regardless of how well the linear models represent the truth.  

\subsection{Preliminary-test procedure}\label{sec:ptp}
In the context of regression analysis of randomized trials,
the {\it preliminary-test procedure} suggests that regression adjustment for covariates should only be conducted if the realized treatment allocation fails to meet some prespecified covariate balance criterion \citep{permutt}.
Let $\hta$ and $\hsea$ denote the covariate-adjusted estimator and its associated standard error based on some prespecified regression specification.
The resulting preliminary-test procedure proceeds with inference in the following two steps:
\begine[(i)]
\item test for covariate balance based on $(X, Z) = (x_i, Z_i)_{i=1}^N$; 
\item conduct inference based on $(\hta, \hsea)$ if the test is rejected, and based on  $(\htn, \hsen)$ if otherwise. 
\ende

Let $\phi = \phi(X, Z) \in \{0,1\}$ denote the corresponding covariate balance test, with $\phi= 0$ if the test is rejected and $\phi = 1$ if otherwise. 
Intuitively, $\phi=1$ indicates that the realized allocation is balanced by meeting the prespecified balance criterion, whereas $\phi = 0$ indicates otherwise. 
The estimator used by the preliminary-test procedure, referred to as the {\it preliminary-test estimator}, is then
\begina
\htpt = \left\{
\begin{array}{ll}
\htn &\text{if $\phi = 1$;}\\
\hta &\text{if $\phi = 0$,}
\end{array}
\right. 
\enda
with standard error
\begina
\hsept = \left\{
\begin{array}{ll}
\hsen &\text{if $\phi = 1$;}\\
\hsea &\text{if $\phi = 0$.}
\end{array}
\right.
\enda
We use the subscript \text{pt} to represent ``preliminary test". The definition of $\phi$ ensures $$\htpt = \phi \cdot \htn + (1-\phi)\cdot \hta, \quad \hsept = \phi\cdot \hsen + (1-\phi)\cdot \hsea.$$  

Recall $\cin$ as the  $1-\alpha$ large-sample confidence interval based on $(\htn, \hsen)$ and normal approximation. Let $\cia = [\hta - \qa  \hsea, \ \hta + \qa  \hsea]$ be the analog based on the covariate-adjusted estimator and its associated standard error. 
Following the same spirit as the definitions of $\htpt$ and $\hsept$, a common strategy is to 
construct the large-sample preliminary-test confidence interval as 
\begina
\cipt =\left\{
\begin{array}{ll}
\cin &\text{if $\phi = 1$;}\\
\cia &\text{if $\phi = 0$.}
\end{array}
\right.
\enda

The definitions of $\hta$ and $\hsea$ above are general and compatible with all common covariate adjustment strategies in practice. 
We categorize the commonly used preliminary-test procedures into two categories depending on whether $(\hta, \hsea)$ adjusts for all or only a subset of the covariates. 

\subsubsection{Preliminary-test procedure without variable selection}
The first category of procedures adjusts for all covariates if the realized allocation fails the balance test. Given the convenient regression-adjusted estimators $(\htf, \hsef)$ and $(\htl, \hsel)$, a straightforward example is to define $(\hta, \cia) = ( \htf, \cif)$ or $(\hta, \cia) = (\htl, \cil)$. The resulting preliminary-test estimator is then $\htpt=\phi \cdot \htn + (1-\phi)\cdot \hts$ for $*\in\{\fisher,\lin\}$, equaling a mixture of $\htn$ and $\hts$. 

Let $
\htx= \hx(1) - \hx(0)$ 
denote the difference in covariate means with $\hx(z) = \meaniz x_i$ for $z = 0,1$. 
Its Mahalanobis distance from the origin,  denoted by 
\begina
M = \htx^\T \{\cov(\htx)\}^{-1}\htx, 
\enda defines a measure of covariate balance across all covariates 
\citep{morgan2012rerandomization}.
A possible test of covariate balance is to consider a treatment allocation as balanced if $M$ is less than a prespecified threshold $a \geq 0$ with  
\begina
\phi = \phim= \mi(M < a).
\enda
\cite{zdrep} showed that such a test is asymptotically equivalent to a multivariate two-sample $t$-test based on $(X, Z) = (x_i, Z_i)_{i=1}^N$.

\subsubsection{Preliminary-test procedure with variable selection}\label{sec:vs}
The second category of procedures  tests balance marginally for each covariate and adjusts for only the significantly imbalanced covariates in  analysis.  A common choice for the covariate-wise balance test is to conduct a two-sample $t$-test for each covariate and declare the covariate imbalanced if  the associated $p$-value is less than a prespecified threshold such as 0.05. 

\bigskip


We will focus on preliminary-test procedures based on $\phim$ without variable selection for the rest of this article due to its conceptual straightforwardness, and derive in the following sections its sampling properties from the randomization-based perspective.
The preliminary-test procedures with variable selection, on the other hand, involve post-selection inference and are hence technically more challenging. 
Despite the limited gain in efficiency from simulation evidence when the number of covariates is small \citep{posthoc, permutt, raab},
we conjecture that such variable-selection procedures may have merits over the always-adjust strategy with high-dimensional covariates, where naive regression adjustment suffers from overfitting.  
We relegate the formal theory to future research.

\black

\section{Preliminary test based on the Mahalanobis distance}\label{sec:pt}
The preliminary-test procedure based on $\phim= \mi(M < a)$ uses the Mahalanobis distance $M$ to measure the covariate balance and conducts regression adjustment if and only if $M$ is greater than or equal to some prespecified threshold $a \geq 0$. 
We focus on the two special cases of $(\hta, \hsea) = (\htf, \hsef)$ and $(\hta, \hsea) = (\htl, \hsel)$ for regression adjustment when $M \geq a$, and derive the sampling properties of the resulting $\htpt$ and $\cipt$ for inferring $\tau$. 
To simplify the presentation, 
let $(\htptf, \hseptf, \ciptf ) $ and $(\htptl, \hseptl, \ciptl ) $ denote the values of $(\htpt, \hsept, \cipt)$  when the analyzer uses the additive and fully interacted specifications for regression adjustment, respectively. We have
\beginy
\label{eq:htpts}
(\htpts, \hsepts, \cipts)= \left\{
\begin{array}{cl}
(\htn, \hsen, \cin) &\text{if $M < a$}\\
(\hts, \hses, \cis) &\text{if $M \geq a$}
\end{array}
\right. 
\endy 
 with $\cipts = [\htpts - \qa  \hsepts, \ \htpts + \qa  \hsepts ]$ for $\fl$.

\subsection{Sampling distributions of the point estimators}
Let 
$\ep \sim \mn(0,1)$ be a standard normal random variable,
$D = (D_1, \ldots, D_J)^\T \sim \mN(0_J,I_J)$ be a $J$-dimensional standard normal random vector, 
and 
\begina
\mL \sim D_1 \mid \{D^\T D < a\}, \quad \ml'  \sim D_1 \mid \{D^\T D \geq a \}
\enda be two truncated normal random variables independent of $\epsilon$. 
Intuitively, the truncated normal distributions characterize the distributions of $\htpts$ when $M<a$ and $M \geq a$, respectively, with $\mL$ first introduced by \cite{LD2018} in the theory of rerandomization.
Let $\pi_a  = \pr( D^\T D <a )$ denote the asymptotic probability of $M < a$ \citep{LD2018}.  For $\nfl$, let $v_*$ denote the asymptotic variance of $\hts$ in the sense that $\sqrt N(\hts - \tau) \rs \mn(0, v_*)$; we give the explicit forms of $v_*$ in the supplementary material. 
Proposition \ref{prop:htpt} below is our first new result in this article and establishes the asymptotic sampling distribution of $\htpts$ as a mixture of two possibly non-normal distributions for $\fl$.

\begin{proposition}\label{prop:htpt}
As $N\to \infty$, we have 
\begina
\sqrt N(\htpts-\tau) &\rs& \mix\beginp
\vl  ^{1/2} \epsilon +  (\vn - \vl)^{1/2}  \ml&: & \pi_a\\
\vl  ^{1/2} \epsilon +  (\vs - \vl)^{1/2}   \ml' &: & 1-\pi_a
\endp
\enda
for $\fl$. 
\end{proposition}

Proposition \ref{prop:htpt} ensures that $\sqrt N(\htpts-\tau)$ converges in distribution to a mixture of $\vl  ^{1/2} \epsilon +  (\vn - \vl)^{1/2}  \ml $ and
$\vl  ^{1/2} \epsilon +  (\vs - \vl)^{1/2}   \ml' $ with weights $(\pi_a, 1-\pi_a)$ for $\fl$. For $\htptl$, the second component reduces to $\vl^{1/2} \epsilon\sim \mn(0, \vl)$ as the coefficient of $\ml'$ equals $0$.

Theorem \ref{thm:htpt} below is a direct consequence of Proposition \ref{prop:htpt} and quantifies the asymptotic relative efficiency of $\htpts \ (\fl)$. 
 
\begin{theorem}\label{thm:htpt}
As $N\to\infty$, 
\begine[(i)]
\item\label{item:htpt1} $\htpts \ (\fl)$ are less efficient than $\htl$; 
\item\label{item:htpt2} $\htptl$ is more efficient than $\htn$ and $\htptf$;
\item\label{item:htpt3} $\htptf$ may be less efficient than $\htn$. 
\ende 
\end{theorem}
Theorem \ref{thm:htpt}\eqref{item:htpt1}, combined with the classical results in Section \ref{sec:review}, ensures that the $\htl$ under the always-adjust strategy is asymptotically the most efficient among $\{\htn, \htf, \htl, \htptf, \htptl\}$. 
Theorem \ref{thm:htpt}\eqref{item:htpt1}--\eqref{item:htpt2} together ensure that the preliminary-test procedure, when using the fully interacted specification, improves efficiency over the unadjusted $\htn$ asymptotically yet is always dominated by $\htl$. 
This is intuitively because of its mixed use of $\htn$ and $\htl$ depending on the covariate balance, where $\htl$ ensures asymptotic efficiency over $\htn$ \citep{Lin13}. 
Theorem \ref{thm:htpt}\eqref{item:htpt3}, on the other hand, suggests that the preliminary-test procedure using the additive specification may be even worse than $\htn$. 
This is intuitively the consequence of $\htf$ being not necessarily more efficient than $\htn$ \citep{Freedman08a}.  
We illustrate these results by simulation in Section \ref{sec:simulation}.

\subsection{Coverage rates of the large-sample confidence intervals} 
Recall that  $
\cipts 
$ 
denotes the $1-\alpha$ large-sample confidence interval based on $(\htpts, \hsepts)$ and normal approximation for $\fl$.
We quantify in this subsection their asymptotic coverage rates for inferring $\tau$. 

Recall that  $
\cis 
$ 
denotes the $1-\alpha$ large-sample confidence interval based on $(\hts, \hses)$ for $\nfl$.  
They are {\it asymptotically valid} in the sense that each $\cis$ covers $\tau$ with probability at least $1-\alpha$ as $N$ tends to infinity and the coverage rates are asymptotically exact if $\tau_i$'s are constant over $i = \ot{N}$.
We call a confidence interval {\it overconservative} if its coverage rate exceeds the nominal level even under the constant treatment effects, and call a confidence interval {\it anticonservative} if its coverage rate can be less than the nominal level. 
Theorem \ref{thm:coverage} below extends the classical results on $\cis \ (\nfl)$ to the preliminary-test intervals and states the asymptotic anti- and overconservativeness of $\ciptf$ and $\ciptl$, respectively. 

\begin{theorem}\label{thm:coverage}
As $N\to\infty$, 
\begine[(i)]
\item\label{item:overconservative_L} $\ciptl$ is overconservative; 
\item\label{item:invalid_F} $\ciptf$ is overconservative when either $e_0=e_1=1/2$ or the individual effects $\tau_i$ are constant across $i = \ot{N}$, but can be anticonservative in the presence of heterogeneous treatment effects and unequal treatment group sizes;  
\item\label{item:width} $\cil$ is the shortest among $\{\cin, \cif, \cil, \ciptf, \ciptl\}$; $\ciptl$ is  no wider than $\cin$, whereas $\ciptf$ can be wider than $\cin$. 
\ende
\end{theorem}

Theorem \ref{thm:coverage}\eqref{item:overconservative_L} justifies the Wald-type inference based on $\ciptl$ but also points out the issue of overconservativeness. 
Intuitively, it follows from the definition in \eqref{eq:htpts} that the coverage rate of $\ciptl$ is a weighted average of the coverage rate of $\cin$ when $M<a$ and the coverage rate of $\cil$ when $M\geq a$:
\begina
\pr(\tau \in \ciptl) \ = \   \pr(\tau \in \cin \mid M < a) \cdot \pr( M < a) + \pr(\tau \in\cil \mid  M \geq a)\cdot \pr( M \geq a).
\enda
Standard theory ensures that $\cil$ is asymptotically valid when $M\geq a$, whereas $\cin$ is asymptotically overconservative when $M<a$. This clarifies the source of overconservativeness in the coverage rate of $\ciptl$. We relegate the details to Theorem \ref{thm:coverage_conditional} in Section \ref{sec:cond} to avoid repetition.

Theorem \ref{thm:coverage}\eqref{item:invalid_F} points out the possible anticonservativeness of $\ciptf$ in the presence of heterogeneous treatment effects such that normal approximation alone is not sufficient for interval estimation based on $(\htptf, \hseptf)$. 
Intuitively, the same reasoning as above ensures
\begina
\pr(\tau \in \ciptf) \ = \  \pr(\tau \in \cin \mid M < a) \cdot \pr( M < a) + \pr(\tau \in\cif \mid  M \geq a)\cdot \pr( M \geq a),
\enda
with $\cin$ being asymptotically overconservative when $M <a$. Our theory in Theorem \ref{thm:coverage_conditional} later nevertheless suggests that 
 $\cif$ can be asymptotically anticonservative when $M\geq a$. This implies that the overall coverage rate of $\ciptf$ may be less than the nominal level asymptotically, especially when $a$ is small such that $\pr(M \geq a)$ dominates the weighted average.
The issue becomes more severe when viewed from a conditional perspective, which we discuss in detail in Section \ref{sec:cond}. 
Two exceptions are when $e_0 = e_1 =1/2$ with equal-sized treatment groups or when $\tau_i$'s are constant across all $i$.
They both ensure that $\cif$ is asymptotically valid when $M \geq a$ such that $\ciptf$ is overall overconservative as a result of the overcoverage of $\cin$ when $M<a$. 
These results complement \cite{permutt}, who studied $\ciptf$ under a joint Gaussian model of outcome and a single covariate and derived its overconservativeness under constant treatment effects. 
We illustrate the novel result on possible anticonservativeness by simulation in Section \ref{sec:simulation}.

Theorem \ref{thm:coverage}\eqref{item:width} suggests that $\ciptl$ is asymptotically narrower than $\cin$ despite the overconservativeness but is wider than $\cil$. 
This, together with the classical results from Section \ref{sec:review}, ensures that $\cil$ is the narrowest among $\{\cin, \cif,\cil,  \ciptl, \ciptf\}$,  ensuring asymptotic validity without overconservativeness. It is our recommendation for interval estimation under complete randomization. 

Corrections for the anti- and overconservativeness of $\cipts \ (\fl)$  can be made by using preliminary test-specific confidence intervals based on the asymptotic distributions in Proposition \ref{prop:htpt}.
The resulting intervals, while ensuring better calibration than $\cipts$, are computationally more complicated and asymptotically still wider than $\cil$. 
We thus do not recommend them and relegate the details to the {\sm}. 

\section{Fisher randomization test}\label{sec:frt}
The discussion so far concerns the point and interval estimation of the average treatment effect $\tau$. 
The Fisher randomization test ({\frt}), on the other hand, provides a convenient way to test 
not only the sharp null hypothesis of no treatment effect on any individual \citep{Fisher35, rubin80, CausalImbens} but also the weak null hypothesis of zero average treatment effect \citep{wd, zdfrt, colin}. 
We extend in this section the discussion to {\frt}s with test statistics from the preliminary-test procedure and evaluate their validity for testing the sharp and weak null hypotheses, respectively.

A test statistic $T = T(Z, Y, X)$ is a function of the treatment assignments, observed outcomes, and  covariate vectors. 
Let $\mathcal Z$ denote the set of all $N!$ permutations of the realized assignment vector $Z$. 
\citet{Fisher35} considered testing the sharp null hypothesis  
\begina
\HF: \tau_i = 0 \quad \text{for all} \ \ i = \ot{N}
\enda
and proposed  the {\frt} to compute the $p$-value as 
\begin{equation}\label{eq:pfrt}
p_\frt  
 = |\mathcal Z|^{-1}\sum_{\bm z \in \mathcal Z} \mi\left\{ \big|T(\bm z, Y, X )\big| \geq \big|T(Z,Y,X)\big|\right\} 
\end{equation}
for a two-sided test.  
Common choices of the test statistic include $\hts$ and $\hts/\hses$ for $\nfl$.

When applied to test statistics based on $(\htpts, \hsepts)$ from the preliminary-test procedure, this means we compute the value of $M$ and $(\hts, \hses) \ (\nfl)$ for each $\bm z \in \mathcal Z$, denoted by $M(\bm z)$ and $(\hts(\bm z), \hses(\bm z))$ respectively, and form $T(\bm z, Y, X)$ using 
\begina
\big( \htpts(\bm z), \hsepts(\bm z) \big)= \left\{
\begin{array}{cll}
\big(\htn(\bm z), \hsen(\bm z) \big) &\text{if $M (\bm z)< a$}\\
\big(\hts(\bm z), \hses(\bm z) \big) &\text{if $M(\bm z) \geq a$}
\end{array}
\right.
\enda
for $\fl$. 
The resulting inference is finite-sample exact for testing $\hf$ regardless of the choice of the test statistic $T$. 
Recall that the confidence interval constructed based on $(\htptf, \hseptf)$ and normal approximation may be asymptotically invalid for inferring $\tau$. The {\frt} for $\hf$ thus remains one way of drawing valid inference from $(\htptf, \hseptf)$  without additional modifications. \cite{edwards} also made this point. 

The nice properties of the {\frt} motivate endeavors to generalize its use to other types of null hypothesis. That is, for an arbitrary null hypothesis $H_0$, we can pretend that we are testing the sharp null hypothesis, and compute $p_\frt$ based on imputation under $\hf$.   We then reject $H_0$ if $p_\frt$ is less than the prespecified significance level \citep{wd}. 
Whereas the finite-sample exactness is in general lost for null hypotheses other than $\hf$, a careful choice of the test statistic can nevertheless deliver asymptotically valid inference based on $\pfrt$. 
Formally, for an arbitrary null hypothesis $H_0$, a test statistic $T$ is {\it asymptotically valid} for testing $H_0$ if under $H_0$,  
$$
\limN  \pr(  p_\frt \leq \alpha   ) \leq \alpha \qquad \text{for all }\alpha\in(0,1).
$$ 
%
%
%
Under this general definition of asymptotic validity, 
\cite{DD18}, \citet{wd}, and \citet{zdfrt} applied the {\frt} to testing the weak null hypothesis 
$$\hn: \tau = 0$$
and showed that the studentized $\hts/\hses \  (\nfl)$ are asymptotically valid, whereas the unstudentized $\hts \ (\nfl)$ are not. 
This illustrates the utility of studentization by robust standard errors in delivering asymptotically valid   {\frt}s for testing $\hn$. 
Theorem \ref{thm:frt} below extends the theory to test statistics based on the \phts procedure and gives a negative result on the validity of $\htpts$ and $\htpts/\hsepts \ (\fl)$ for testing $\hn$.

\begin{theorem}\label{thm:frt}
None of $\htpts$ and $\htpts/\hsepts \ (\fl)$ are asymptotically valid for testing $\hn$. 
\end{theorem}

Theorem \ref{thm:frt} suggests that studentization by robust standard errors alone is no longer sufficient to ensure the asymptotic validity of test statistics from the preliminary-test procedure for testing $\hn$.  
This is a consequence of the asymptotic non-normality of $\htpts \ (\fl)$ and illustrates the complications incurred by the preliminary-test procedure on subsequent analysis.
We illustrate the possibly liberal type I error rates by simulation in Section \ref{sec:simulation}.

In cases where inference based on the hypothesis testing of $\hn$ is still desired for outputs from the preliminary-test procedure, we recommend using the prepivoting procedure proposed by \cite{colin} to restore the asymptotic validity of $\htpts$ and $\htpts/\hsepts \ (\fl)$.
The resulting {\frt}s simultaneously deliver finite-sample exact inference for $\hf$ and asymptotically valid inference for $\hn$.
We relegate the details to the {\sm}. 

The preliminary-test intervals $\ciptf$ and $\ciptl$ give another approach to testing $\hn$, recalling the duality between hypothesis testing and interval estimation. By Theorem \ref{thm:coverage},  hypothesis testing based on $\ciptl$ preserves the type I error rates asymptotically, whereas that based on $\ciptf$ does not. 
This differs from Theorem \ref{thm:frt}, illustrating the discrepancy between the two approaches in case of non-normal asymptotic distributions.
In addition, the prepivoted {\frt}s are finite-sample exact for the sharp null hypothesis whereas the confidence interval-inverted tests are not.

\section{A trio of covariate adjustment}\label{sec:trio}
Sections \ref{sec:pt} and \ref{sec:frt} clarify the sampling properties of the \phts procedure from the randomization-based perspective.
In particular, 
the confidence interval $\ciptf$ based on \cite{Fisher35}'s analysis of covariance and normal approximation may be anticonservative for interval estimation,
and studentization alone can no longer ensure the asymptotic validity of $\htpts / \hsepts \ (\fl)$ for testing the weak null hypothesis.
This illustrates the complications incurred by the preliminary-test procedure on subsequent inference. 
We now put these results in perspective, and unify the preliminary-test procedure with regression adjustment and rerandomization as two prevalent approaches to handling covariate imbalance in randomized trials.

\subsection{Regression adjustment for covariate imbalance}

Regression adjustment by \cite{Fisher35} and \cite{Lin13} can be viewed as adjusting for covariate imbalance in the analysis stage.
Specifically, let $\hgf$ be the coefficient vector of $x_i$ from $\lmt(Y_i \sim 1+Z_i +x_i)$  over $i = \ot{N}$, and  let $\hg_\lin =e_0\hglo+ e_1\hglz$, where $\hglzz$ denotes the coefficient vector of $x_i$ from $\lmt(Y_i \sim 1 + x_i)$ over $\{i: Z_i = z\}$. 
\citet[][Proposition 1]{zdfrt} showed that
\begina
\hts =  \htn - \hgs^\T  \htx \qquad (\fl).
\enda
The two covariate-adjusted estimators are thus effectively variants of $\hat{\tau}_\nm $ with corrections for the difference in covariate means.

\subsection{Rerandomization based on the Mahalanobis distance}
Rerandomization, on the other hand, enforces covariate balance in the design stage \citep{morgan2012rerandomization}, and accepts an allocation if and only if it satisfies some prespecified covariate balance criterion.
\cite{morgan2012rerandomization} and \cite{LD2018} studied a special type of rerandomization, known as ReM, that uses the Mahalanobis distance $M$ as the balance criterion and accepts a randomization if and only if $M < a$ for some prespecified threshold $a \geq 0$. 
It can be viewed as a 
more proactive alternative of the preliminary-test procedure, using results from the covariate balance test to guide not downstream analysis but upstream design, preventing covariate imbalance from the very beginning. 

\cite{LD20} established the duality between ReM and regression adjustment for improving efficiency under the treatment-control trial, and showed that ReM does not affect the asymptotic efficiency of the efficiently adjusted $\htl$. 
This justifies the large-sample Wald-type inference based on $\cil$ under ReM, identical to that under complete randomization.
\cite{zdfrt} further clarified the impact of ReM on the {\frt}s for both $\hf$ and $\hn$. In particular, they showed that $\htn/\hsen$ is no longer asymptotically valid for testing $\hn$ under ReM, whereas the covariate-adjusted $\hts/\hses \ (\fl)$ still are. This highlights the value of covariate adjustment for valid {\frt}s of $\hn$ under ReM, and provides the intuition behind the asymptotic invalidity of $\htpts/\hsepts \ (\fl)$   in Theorem \ref{thm:frt} as a result of the asymptotic invalidity of $\htn/\hsen$ when $M<a$.  See \cite{zach3} for a comparison of testing power under ReM and complete randomization.

\subsection{A unified look and final recommendation}
\begin{table}[t]\caption{\label{tb:summary}Impact of regression adjustment, the preliminary-test procedure, and ReM on subsequent inference. Assume $\phim = \mi(M <a)$ as the covariate balance criterion for both the preliminary-test procedure and ReM. This ensures $\htpts = \htn$ under ReM for $\fl$.}

\renewcommand{\arraystretch}{1.2}
\begin{center}
\begin{tabular}{c|c|cc }\hline
&& Asymptotic sampling & Validity of CI based \\
Procedure&Estimator& distribution of $\sqrtn (\htau - \tau)$ & on $\htau \pm \qa  \hse$ \\\hline
Unadjusted &$\htn$ & $\mn(0, \vn)$ &yes \\\hline
\multirow{2}{*}{Regression adjustment}&$\htf$ & $\mn(0, \vf)$ &yes \\
 &$\htl$ & $\mn(0, \vl)$ &yes \\\hline
\multirow{3}{*}{Preliminary test } & $\htptf$ & $\mix\beginp
\vl  ^{1/2} \epsilon +  (\vn - \vl)^{1/2}  \ml&: & \pi_a\\
\vl  ^{1/2} \epsilon +  (\vf - \vl)^{1/2}   \ml' &: & 1-\pi_a
\endp $ & no \\
 & $\htptl$ &$\mix\beginp
\vl  ^{1/2} \epsilon +  (\vn - \vl)^{1/2}  \ml&: & \pi_a\\
\vl  ^{1/2} \epsilon&: & 1-\pi_a
\endp $  & overconservative \\\hline
&$\htn $ & $ \vl  ^{1/2} \epsilon +  (\vn-\vl)^{1/2}  \ml$ & overconservative \\
ReM &$ \htf  $ &$  \vl  ^{1/2} \epsilon +  (\vf-\vl)^{1/2}  \ml$ &   overconservative \\
&$ \htl $ &  $\mn(0, \vl)$ & yes  \\\hline
\end{tabular}
\end{center}

\end{table}

Table \ref{tb:summary} summarizes the impact of regression adjustment, the preliminary-test procedure, and ReM on subsequent point and interval estimation. 
The results on regression adjustment and ReM are not new, whereas those on the preliminary-test procedure are a summary of our new results from the previous sections. 
For point estimation, $\hts \  (\nfl)$ and $\htpts \ (\fl)$ are all consistent under complete randomization, with $\htl$ being asymptotically the most efficient.
The preliminary-test procedure improves the efficiency over $\htn$ if and only if we use the fully interacted specification for regression adjustment when $M \geq a$. 
ReM retains the consistency of $\hts \ (\nfl)$ and improves the asymptotic efficiency of $\htn$ and $\htf$. They are nevertheless still asymptotically less efficient than $\htl$, whose asymptotic efficiency is unaffected by ReM.
This illustrates the superiority of $\htl$ for point estimation under both complete randomization and ReM. 

For interval estimation, the
$\cis$'s $(\nfl)$ based on $(\hts, \hses)$ and normal approximation are asymptotically valid under complete randomization, with $\cil$ having the smallest width.  
ReM has no effect on $\cil$ asymptotically yet renders $\cin$ and $\cif$ overconservative. 
The preliminary-test procedure, on the other hand, yields $\ciptl$ that is overconservative and $\ciptf$ that is possibly anticonservative. 
Among $\{\cin, \cif, \cil, \ciptl \}$ that ensure the nominal coverage rates asymptotically, 
the $\cil$ based on the fully interacted specification has the smallest asymptotic width without overconservativeness under both complete randomization and ReM.  
It is our recommendation for interval estimation under both complete randomization and ReM.

For hypothesis testing based on the {\frt}, all eight estimators are finite-sample exact for testing $\hf$ but not asymptotically valid for testing $\hn$. 
Table \ref{tb:summary_robust} further summarizes the utility of studentization for ensuring the asymptotic validity under $\hn$. 
Compared with their unstudentized counterparts, 
studentization restores the asymptotic validity of 
$\hts/\hses \ (\nfl)$ for testing $\hn$ under complete randomization and that of $\hts / \hses  \ (\fl)$ under ReM. 
The $\htn/\hsen$ under ReM and the $\htpts/\hsepts \ (\fl)$ under the preliminary-test procedure, on the other hand, remain invalid even after studentization. 
This illustrates the complications incurred by the preliminary-test procedure and ReM on the {\frt}. 

\begin{table}[t]\caption{\label{tb:summary_robust}Impact of studentization on the {\frt}. Assume $\phim = \mi(M <a)$ as the covariate balance criterion for both the preliminary-test procedure and ReM.  This ensures $\htpts = \htn$ under ReM for $\fl$.
}
\begin{center}\renewcommand{\arraystretch}{1.2}
\begin{tabular}{cccc}\hline
&&  Finite-sample  & Asymptotically  \\
 Procedure &Test statistic & exact  for $\hf$ &valid for $\hn$\\\hline
Unadjusted &$\htn/\hsen$ & yes &  yes \\\hline
\multirow{2}{*}{Regression adjustment} &$\htf/\hsef$ & yes &  yes \\
 &$\htl/\hsel$ & yes &  yes\\\hline
\multirow{2}{*}{Preliminary test} & $\htptf/\hseptf$ &  yes & no\\
 & $\htptl / \hseptl$ &yes & no\\\hline
 &$\htn/\hsen $ & yes &  no \\
ReM &$ \htf/\hsef $ & yes &  yes \\
&$\htl/\hsel  $ & yes &  yes \\\hline
\end{tabular}
\end{center}
\end{table}

Juxtapose Tables \ref{tb:summary} and \ref{tb:summary_robust}. Regression adjustment by the fully interacted specification dominates the preliminary-test procedure and ReM in terms of both asymptotic efficiency for estimation and asymptotic validity for testing $\hn$. 

Despite being dominated by the fully interacted adjustment in terms of asymptotic efficiency, however, ReM improves the actual balance in realized allocations, which is of great value in terms of (i) simplifying the interpretation of experimental results, (ii) closing the gap between the adjusted and unadjusted estimators \citep[Section S5.3]{zdca}, and (iii) reducing conditional bias \citep[Section 4.3]{vif}, to name just a few. We thus recommend the combination of ReM and fully interacted adjustment for both estimation and hypothesis testing based on the {\frt}.

The preliminary-test procedure, on the other hand, incurs substantial complications on the usual inferences based on normal approximation and studentization. 
In particular, the $\htptf$ based on the additive specification is asymptotically less efficient than $\htl$, and not valid for either constructing confidence interval based on normal approximation or testing $\hn$ by the {\frt} even after studentization. 
The $\htptl$ based on the fully interacted specification is asymptotically less efficient than $\htl$, overconservative for constructing confidence interval based on normal approximation, and invalid for testing $\hn$ by the {\frt} even after studentization. 
We thus do not recommend the preliminary-test procedure for either estimation or hypothesis testing based on the {\frt}.

\section{A conditional perspective}\label{sec:cond}
Recall that \cite{begg} and \cite{senn} argued critically against the logical foundation of preliminary tests.
Both authors, nevertheless, acknowledged the value of covariate balance information in yielding more informed inference {\it conditionally}.
Intuitively, the conditional perspective evaluates the sampling properties of the procedure of interest over repeated samples that are {\it similar} to the one that is actually realized, promising more relevant quantification of the uncertainty with regard to the data at hand.
We extend in this section their discussion to conditional inference based on the preliminary-test procedure under the randomization-based framework, 
and 
consider the conditional properties of the preliminary-test procedure based on $\phim$.
The resulting theory, as it turns out, provides important building blocks for understanding the results under unconditional inference in Sections \ref{sec:pt}--\ref{sec:trio}. 

\subsection{Similarity of allocations based on covariate balance test}
The discussion in Sections \ref{sec:pt}--\ref{sec:trio} takes an unconditional perspective, and evaluates the sampling properties of the preliminary-test procedure over all possible complete randomizations of treatment sizes $(N_0, N_1)$. 
The conditional perspective instead focuses only on the subset of randomizations that are {\it similar} to the one that is actually realized. 

Recall that $\mz$ denotes the set of all possible permutations of the assignment vector $Z$. 
Its nonrepetitive elements constitute all possible complete randomizations of treatment sizes $(N_0, N_1)$. 
Motivated by the way the preliminary-test procedure determines the inference specification by the value of $\phim$, we define $\phim(\bm z)$ as the value of $\phim$ if $Z= \bm z$ for each $\bm z \in \mz$, and view two allocations $\bm z$ and $\bm z'$ as similar if $\phim(\bm z) = \phim(\bm z')$.
This divides all possible complete randomizations into two types: those that are balanced with $\phim(\bm z) = 1$ and those that are unbalanced with $\phim(\bm z) = 0$. 

The conditional inference then evaluates the sampling properties of $\htpts$ and $\cipts$ over all allocations that are similar to the realized $Z$ in the sense of $\phim(\bm z) = \phim$. 
Intuitively, this means that we evaluate the distribution of $\htpts$ and the coverage rates of $\cipts$ over all balanced allocations with $\phim(\bm z) = 1$ if the realized allocation is balanced with $\phim = 1$, and over all unbalanced allocations with $\phim(\bm z) = 0$ if the realized allocation is unbalanced with $\phim = 0$.

\subsection{Sampling distributions  conditioning on $\phi_\textup{M}$}
Proposition \ref{prop:cond} below states the asymptotic sampling distributions conditioning on $\phim$. 

\begin{proposition}\label{prop:cond}
\begine[(i)]
\item Conditioning on $\{\phim = 1\}$,  we have $\htpts=\htn$  for $\fl$, with 
\begina
\rtn  ( \htpts - \tau)  \mid \left\{\phim = 1 \right\}&\rs&    \vl  ^{1/2} \epsilon +  (v_\nm -\vl)^{1/2} \ml \quad(\fl).
\enda
\item\label{item:htl_cond} Conditioning on $\{\phim = 0\}$,  we have $\htpts=\hts$  for $\fl$, with 
\begina
\rtn  ( \htpts - \tau)  \mid \left\{\phim = 0 \right\}&\rs&   \vl  ^{1/2} \epsilon +  (v_*-\vl)^{1/2} \ml' \quad(\fl). 
\enda
\ende
\end{proposition}

Proposition \ref{prop:cond} establishes the asymptotic sampling distribution of $\htpts$ as a convolution of independent normal and truncated normal random variables conditioning on $\phim$.
The forms of the asymptotic distributions in Proposition \ref{prop:cond}(i) follow from the numerical equivalence $\htpts=\htn$ given $\phim =1$ and the asymptotic conditional distribution of $\htn$ given $\phim =1$ established by \cite{LD2018}. Similarly, the forms of the asymptotic distributions in Proposition \ref{prop:cond}(ii) follow from the numerical equivalence $\htpts=\hts$ given $\phim =0$ and the asymptotic conditional distributions of $\hts \ (\fl)$ given $\phim =0$ as a novel result.
The truncation is intuitively the result of the restrictions imposed by $\phim = 1$ and $\phim = 0$, respectively. 
Juxtapose Proposition \ref{prop:cond} with Proposition \ref{prop:htpt}.
The asymptotic unconditional distribution of $\htpts$ in Proposition \ref{prop:htpt} is essentially a mixture of the asymptotic conditional distributions of $\htpts$ given $\phim = 1$ and $\phim = 0$, respectively, with $\pia$ giving the asymptotic probability of $\phim = 1$. 
This gives the intuition behind Proposition \ref{prop:htpt} from a conditional perspective.

Further observe that Proposition \ref{prop:cond}\eqref{item:htl_cond} ensures that $\sqrtn(\htl - \tau)$ is asymptotically $\mn(0, \vl)$ given  $\phim =0$, identical to its asymptotic distribution without conditioning. 
This, together with $\sqrtn (\htl - \tau) \mid \{\phim = 1\} \rs \mn(0, \vl)$ by \cite{LD20}, ensures that the covariate balance status has no effect on the efficiently adjusted $\htl$ asymptotically, echoing the duality between regression adjustment and ReM in improving efficiency \citep{LD20}.

Theorem \ref{thm:cond} below follows from Proposition \ref{prop:cond} and quantifies the conditional relative efficiency of $\htpts \ (\fl)$.

\begin{theorem}\label{thm:cond}
\begine[(i)]
\item Conditioning on $\{\phim = 1\}$,  $\htpts = \htn \ (\fl)$ are asymptotically less efficient than $\htl$. 
\item Conditioning on $\{\phim = 0\}$, $\htptl=\htl$ whereas $\htptf =\htf$ is asymptotically less efficient than $\htl$. 
\ende
\end{theorem}

Recall the inferiority of $\htpts \ (\fl)$ to $\htl$ for unconditional inference from Theorem \ref{thm:htpt}.
Theorem \ref{thm:cond} ensures that such inferiority extends to the conditional inference given $\phim$, corroborating the advantage of the always-adjust strategy over the preliminary-test procedure. 
Intuitively, the results follow from the numerical equivalence of $\htpts$ to $\htn$ and $\hts \ (\fl)$ given $\phim = 0$ and $\phim=1$, respectively, and the conditional inferiority of $\htn$ and $\htf$ to $\htl$ as established by \cite{LD20}.

We focus on the conditional inference given $\phim$ due to its connection with the preliminary-test procedure. 
Alternatively, we can also consider the conditional inference given $\htx$, which leaves $\htl$ the only consistent estimator among $\{\htn, \htf, \htl, \htptf, \htptl\}$. 
This extends the discussion of \cite{senn89}, who formalized the notion of {\it conditional size} as the probability of a type I error given a particular value of $\htx$ and gave its explicit form under a joint Gaussian model of outcome and a single covariate with known dispersion parameters. See also \cite{begg},  \cite{senn}, and \cite{vif}.

\subsection{Coverage rates conditioning on $\phi_{\textup{M}}$}
Recall that $\cipts \ (\fl)$ denote the $1-\alpha$ large-sample  confidence intervals based on $(\htpts, \hsepts)$  and normal approximation. 
Theorem \ref{thm:coverage} evaluates their asymptotic  coverage rates over all possible  complete randomizations of treatment group sizes $(N_0, N_1)$. 
Theorem \ref{thm:coverage_conditional} below instead clarifies their respective conditional coverage rates over all allocations with the same covariate balance status as the realized $Z$. The results, in addition to being informative in themselves, also provide the building blocks for understanding Theorem \ref{thm:coverage}.

\begin{theorem}\label{thm:coverage_conditional}
\begine[(i)] 
\item Conditioning on $\{\phim = 1\}$, we have $\cipts = \cin \ (\fl)$ are both asymptotically overconservative.

\item 
Conditioning on $\{\phim = 0\}$, we have $\cipts = \cis \ (\fl)$ with $\ciptl $ being asymptotically valid. The $\ciptf$ is asymptotically valid if either $e_0 = e_1 = 1/2$ or $\tau_i$'s are constant across $i = \ot{N}$, but can be asymptotically anticonservative if otherwise. 
\ende
\end{theorem}

Theorem \ref{thm:coverage_conditional} follows from Proposition \ref{prop:cond} and clarifies the possible over- and anticonservativeness of $\cipts \ (\fl)$ conditioning on $\phim$. 
Recall that the unconditional coverage rate is essentially a weighted average of the conditional coverage rates given $\phim = 1$ and $\phim = 0$, respectively. 
These conditional over- and anticonservativeness thereby elucidate the source of the unconditional over- and anticonservativeness in Theorem \ref{thm:coverage}.

In particular, Proposition \ref{prop:cond} ensures that $\sqrtn (\htpts -\tau) \ (\fl)$ converge in distribution to $\vl  ^{1/2} \epsilon +  (\vn-\vl)^{1/2}  \ml$ given $\phim = 1$, which has narrower central quantile ranges than $\mn(0, \vn)$ unless $\vn = \vl$ \citep{LD2018, LD20}. 
This underlies the overconservativeness of $\cipts = \cin$ based on normal approximation given $\phim = 1$. 

Similarly, Proposition \ref{prop:cond} ensures that $\sqrtn (\htptf -\tau)$ converges in distribution to $\vl  ^{1/2} \epsilon +  (\vf-\vl)^{1/2}  \ml'$ given $\phim = 0$, which has wider central quantile ranges than $\mn(0, \vf)$ unless $\vf = \vl$. 
This underlies the possible anticonservativeness of $\ciptf = \cif$ based on normal approximation when $\phim = 0$. 
Two exceptions are when $e_0 = e_1 = 1/2$ or $\tau_i$'s are constant such that $\vf = \vl$. The resulting $\ciptf =\cif$ is asymptotically valid given $\phim = 0$. 

Intuitively, balanced allocations with $\phim(\bm z) = 1$ narrow the asymptotic distributions of $\htpts = \htn \ (\fl)$ such that the confidence intervals based on normal approximation overcover. 
In contrast, unbalanced allocations with $ \phim(\bm z) = 0$ widen the asymptotic distribution of $\htptf = \htf$ such that the confidence interval based on normal approximation undercovers.
This provides the intuition behind Theorems  \ref{thm:coverage} and \ref{thm:coverage_conditional}.

\subsection{Conditional FRT}
Recall the setting of the standard {\frt} in Section \ref{sec:frt}.
The original proposal by \cite{Fisher35} computes $\pfrt$ by \eqref{eq:pfrt} using all possible permutations of $Z$.
The conditional {\frt}, in contrast, computes the $p$-value using only allocations that are {\it similar} to $Z$ by some prespecified criterion \citep{jonathan, zach}. 
Under the definition of similarity by covariate balance test, this means 
that we permute $Z$ over all possible allocations $\bm z$ with $\phim(\bm z) = \phim$, denoted by 
\begina
\mz(\phim) = \{\bm z: \phim(\bm z) = \phim\},
\enda
and compute the resulting $p$-value as 
\begin{equation}\label{eq:p_cond}
p_{\frt}(\phim)
 = |\mathcal Z(\phim)|^{-1}\sum_{\bm z \in \mathcal Z(\phim)} \mi\left\{ \big|T(\bm z, Y, X )\big| \geq \big|T(Z,Y,X)\big|\right\}
\end{equation}
for a two-sided test given the observed value of $\phim$.
Theorem \ref{thm:frt_cond} below evaluates the validity of $\pfrt(\phim)$ for testing $\hf$ and $\hn$.

\begin{theorem}\label{thm:frt_cond}
\begine[(i)]
\item Conditioning on $\{\phim = 1\}$, we have $\htpts/\hsepts = \htn/\hsen \ (\fl)$ are both finite-sample exact for testing $\hf$, but neither are asymptotically valid for testing $\hn$.
\item Conditioning on $\{\phim = 0\}$, we have $\htpts/\hsepts = \hts /\hses \ (\fl)$ are both finite-sample exact for testing $\hf$, but only $\htptl /\hseptl$ is asymptotically valid for testing $\hn$.
\ende 
\end{theorem}

Juxtapose Theorem \ref{thm:frt_cond} with Theorem \ref{thm:frt}. 
The conditional {\frt} restores the asymptotic validity of $\htptl / \hseptl$ for testing $\hn$ when $\phim = 0$. 
This is intuitively the consequence of $(\htptl, \hseptl) = (\htl, \hsel)$ and the asymptotic normality of $\htl$ given $\phim = 0$. The $\htptf/ \hseptf$  nevertheless remains invalid for $\hn$ in both cases even asymptotically.
This is intuitively due to the asymptotic non-normality of $\htn$ and $\htf$ given $\phim = 1$ and $\phim = 0$, respectively. We relegate the details to the {\sm}. 

This concludes our discussion on the conditional inference based on the preliminary-test procedure. The conditional perspective restores the asymptotic validity of $\ciptl$ for interval estimation and $\htptl / \hseptl$ for testing $\hn$ given $\phim = 0$, yet leaves all other problems with the preliminary-test procedure unsolved.

\section{Simulation}\label{sec:simulation}

We now illustrate the finite-sample properties of the \phts procedure by simulation.
The results are coherent with the asymptotic theory in Theorems \ref{thm:htpt}--\ref{thm:frt}, featuring the overall efficiency of $\htl$ over $\htpts \ (\fl)$, the possible anticonservativeness of $\ciptf$, and the invalidity of $\htpts$ and $\htpts/\hsepts \ (\fl)$ for testing $\hn$ by the {\frt}. 

\subsection{Efficiency of point estimation}

Consider a finite population of $N=500$ units, $i = \ot{N}$, subject to a treatment-control trial of treatment sizes $(N_0,N_1) = (400, 100)$. 
For each $i$, we draw a $J=5$ dimensional covariate vector $x_i =  (x_{i1}, \dots, x_{i5})^\T$ with $x_{ij}$ as independent Uniform$(-1,1)$, and generate the potential outcomes as  
$Y_i(0) \sim \mathcal{N}( - \sum_{j=1}^5 x^3_{ij} , 0.1^2)$ and $Y_i(1) \sim \mathcal{N}( \sum_{j=1}^5 x^3_{ij}  , 0.4^2)$. 

Fixing $\{Y_i(0), Y_i(1), x_i\}_{i=1}^N$ in the simulation, we draw a random permutation of $N_1$ 1's and $N_0$ 0's  to obtain the allocation under complete randomization, and compute the point estimators under the unadjusted, always-adjust, and preliminary-test strategies, respectively.
%
%
%
Figure \ref{fig:2} shows the distributions of $\hts  -\tau  \ (\ms)$ and $\htpts  -\tau  \ (\fl)$  over $10000$ independent allocations.
For comparison, we also include the distributions of $\hts - \tau \ (\ms)$ under ReM, which are summarized over the subsets of allocations that satisfy the balance criterion $M < a$.
We vary the threshold $a$ for balancedness from the 20\% percentile of the $\chi_5^2$ distribution, which considers approximately $20\%$ of the allocations as balanced, to  the 80\% percentile of the $\chi_5^2$ distribution, which considers approximately $80\%$ of the allocations as balanced. 
The message is coherent across the two thresholds: neither $\htptf$ nor $\htptl$ is as efficient as $\htl$, whereas $\htptl$ improves efficiency over $\htn$, both in line with Theorem \ref{thm:htpt}. 
ReM, on the other hand, has no effect on $\htl$ but improves the efficiency of $\htn$ and $\htf$. 
Compare subplots (a) and (b) under different values of the threshold $a$. 
Increasing the threshold decreases the efficiency of $\htptl$ and $\hts \ (\nf)$ under ReM when everything else stays the same. 

\begin{figure}

\centering 

 \begin{subfigure}{0.8\textwidth}
\includegraphics[width =\textwidth]{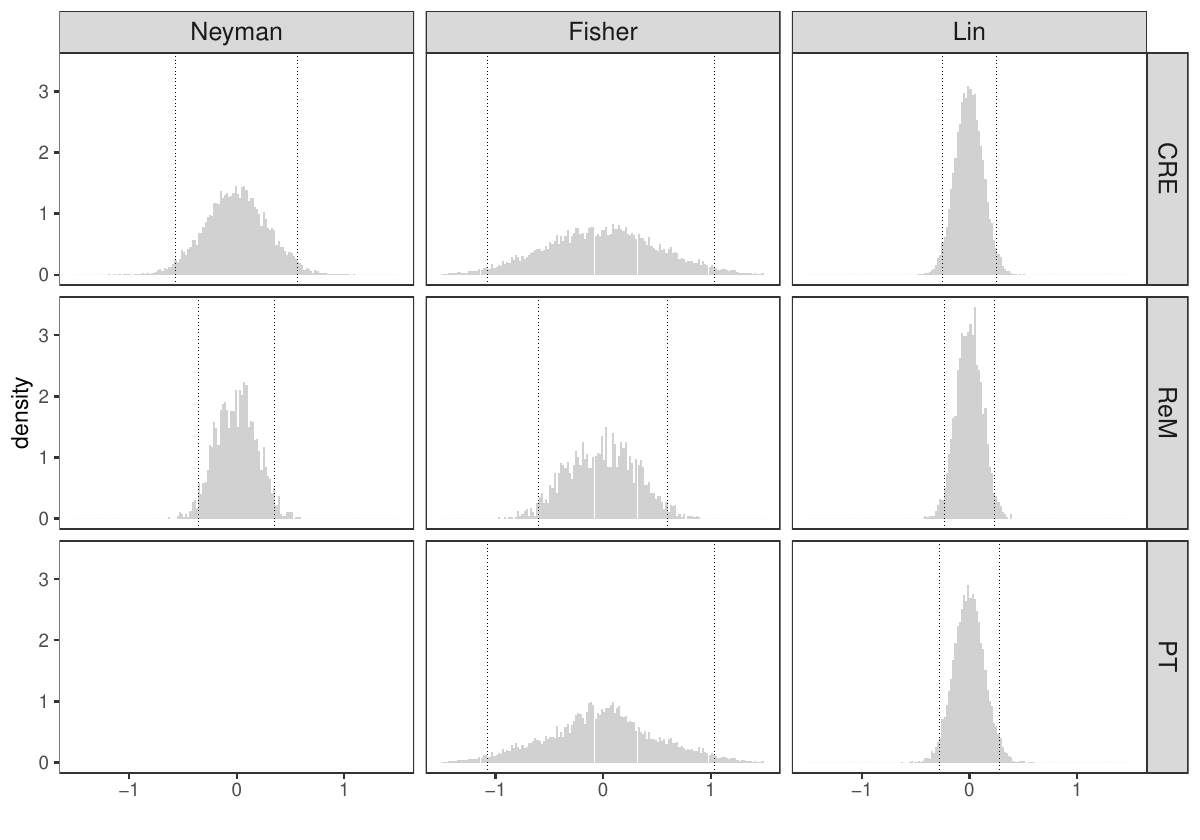}
\caption{$a$ equals the  20\% percentile of the $\chi^2_5$ distribution.}
\label{fig:a}
\end{subfigure}

\bigskip
\bigskip

\begin{subfigure}{0.8\textwidth}
\includegraphics[width =\textwidth]{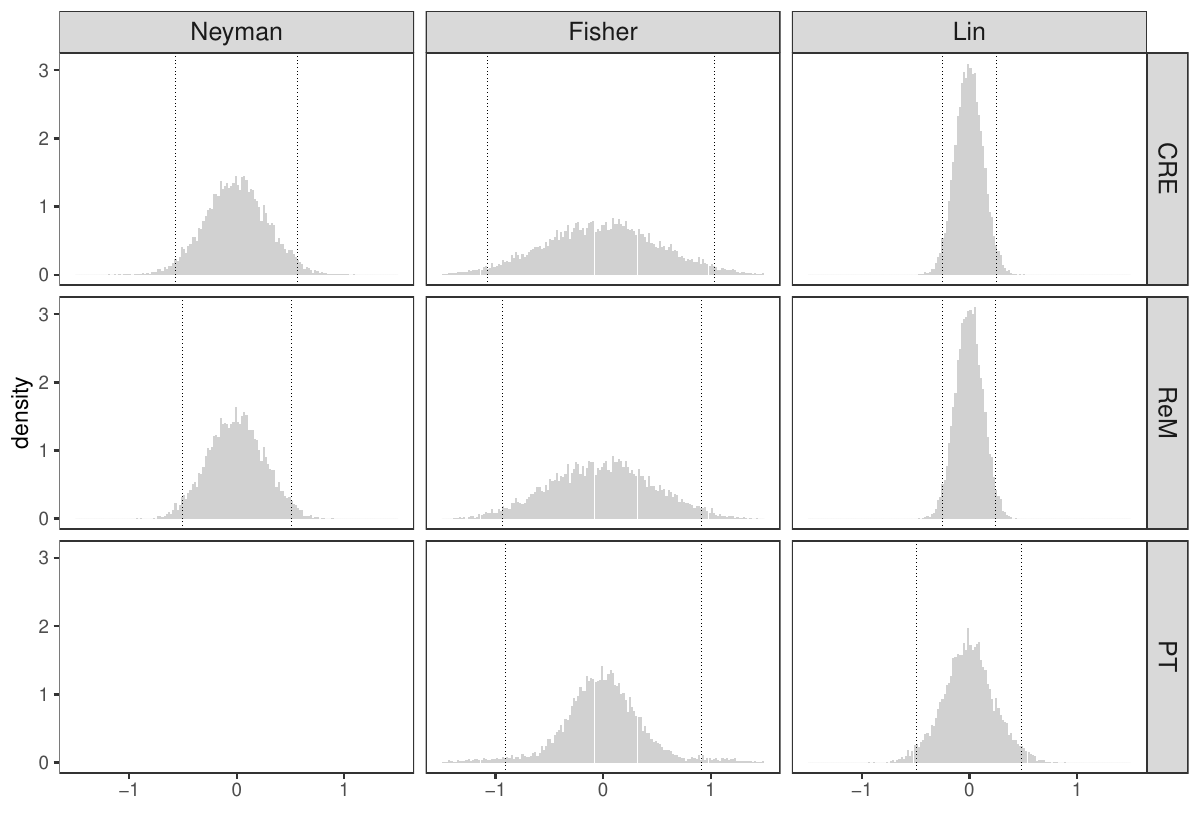}
\caption{$a$ equals the   80\% percentile of the $\chi^2_5$ distribution.}
\label{fig:b}
\end{subfigure}

\bigskip

\caption{\label{fig:2}Distributions of $\hts  - \tau \ (\ms)$ and $ \htpts - \tau \ (\fl)$  over $10000$ independent complete randomizations. 
The vertical lines correspond to the $0.025$ and $0.975$ empirical quantiles, respectively.}
\end{figure}

\subsection{Coverage rates of the confidence intervals}

We illustrate in this subsection the possible anticonservativeness of $\ciptf$. 
Consider a finite population of $N=2000$ units, $i = \ot{N}$, subject to a treatment-control trial of treatment sizes $(N_0, N_1) = (1900, 100)$. 
We generate the covariates $x_i$'s as independent $\mn(0,1)$, and the potential outcomes as 
$
Y_i(0) = -2.5 x_i + \ep_i$ and $Y_i(1) = x_i + \ep_i$,
where $\ep_2, \ldots, \ep_N$ are independent $\mn(0,\sigma_\ep^2)$ and $\ep_1$ is chosen to ensure $\sumi \ep_i x_i = 0$.  Trials and errors suggest that such ground truths render the resulting $\cif$ more likely to undercover when $\phim =0$, which in turn increases the possibility of $\ciptf$ being anticonservative based on the comments after Theorem \ref{thm:coverage}. We vary $\sigma_\ep$ in the simulation to illustrate that the anticonservativeness can occur over a range of finite populations. 

\renewcommand{\arraystretch}{1.2}
\begin{table}[t]\caption{\label{tb:cr}Coverage rates of $\cin$, $\cif$, and $\ciptf$ over $10000$ independent complete randomizations}
\centering
(a) Anticonservative $\ciptf$'s. 

\smallskip

\begin{tabular}{cc|ccc|cc|cc}
  \hline
  && \multicolumn{3}{c|}{Unconditional} & \multicolumn{2}{c|}{$\phim = 1$} & \multicolumn{2}{c}{$\phim = 0$} \\
  $\pia$ & $\sigma_\ep$ & $ \cin$ & $ \cif$& $\ciptf$ & $ \cin =\ciptf$ & $\cif$ & $\cin$ & $\cif = \ciptf$  \\
  \hline
0.75 & 1.5 & 96.8 & 95.4 & 94.0 & 98.0 & 99.8 & 93.3 & 81.9  \\ 
  0.75 & 2.0 & 96.3 & 95.5& 93.7 & 97.2 & 99.6 & 93.7 & 83.1  \\ 
  0.75 & 2.5 & 96.1 & 95.5& 93.7 & 96.8 & 99.1 & 94.0 & 84.6  \\ \hline
  0.80 & 1.5 & 96.8 & 95.4 & 93.9 & 97.8 & 99.7 & 92.7 & 78.0  \\ 
  0.80 & 2.0 & 96.3 & 95.5 & 93.6 & 97.1 & 99.4 & 93.3 & 79.8  \\ 
  0.80 & 2.5 & 96.1 & 95.5& 93.8 & 96.6 & 98.7 & 93.8 & 82.3  \\ \hline
  0.85 & 1.5 & 96.8 & 95.4 & 93.9 & 97.7 & 99.5 & 91.7 & 72.1  \\ 
  0.85 & 2.0 & 96.3 & 95.5 & 93.7  & 97.0 & 99.0 & 92.5 & 75.5\\ 
  0.85 & 2.5 & 96.1 & 95.5& 93.9 & 96.6 & 98.4 & 93.2 & 78.8  \\ 
   \hline
\end{tabular}

\bigskip

(b) $\ciptf$'s that ensure the nominal coverage rate. 

\smallskip

\begin{tabular}{cc|ccc|cc|cc}
  \hline
  && \multicolumn{3}{c|}{Unconditional} & \multicolumn{2}{c|}{$\phim = 1$} & \multicolumn{2}{c}{$\phim = 0$} \\
  $\pia$ & $\sigma_\ep$ & $ \cin$ & $ \cif$& $\ciptf$ & $ \cin =\ciptf$ & $\cif$ & $\cin$ & $\cif = \ciptf$  \\
  \hline
0.05 & 0.5 & 98.5 & 95.2 & 95.2 & 100.0 & 100.0 & 98.5 & 94.9  \\ 
0.05 & 1.0 & 97.6 & 95.4& 95.3 & 99.6 & 100.0 & 97.5 & 95.1  \\ \hline
  0.50 & 0.5 & 98.5 & 95.2 & 95.2  & 100.0 & 100.0 & 97.0 & 90.3 \\ 
  0.50 & 1.0 & 97.6 & 95.4 & 95.0 & 99.2 & 100.0 & 95.9 & 90.7  \\ \hline
  0.95 & 0.5 & 98.5 & 95.2 & 95.6 & 99.6 & 99.2 & 79.1 & 24.8  \\ 
  0.95 & 1.0 & 97.6 & 95.4  & 95.0 & 98.4 & 98.8 & 82.8 & 35.9 \\ 
   \hline
\end{tabular}
\end{table}

Fixing $\{Y_i(0), Y_i(1), x_i\}_{i=1}^N$ in the simulation, we draw a random permutation of $N_{1}$ 1's and $N_{0}$ 0's to obtain the observed outcomes and construct the confidence intervals based on normal approximation.

Table \ref{tb:cr} shows the coverage rates of $\cin$, $\cif$, and $\ciptf$ at confidence level $1-\alpha = 0.95$ over $10000$ independent allocations at different combinations of $(a, \sigma_\ep)$.
In particular, Table \ref{tb:cr}(a) summarizes the cases where $\ciptf$ is anticonservative over a specific range of $a$.
The undercoverage by $\ciptf$ is obvious both unconditionally and conditionally given $\phim = 0$.
Table \ref{tb:cr}(b) shows other cases where $\ciptf$ ensures the nominal coverage rate.
Despite the substantial conditional undercoverage of $\ciptf = \cif$ given $\phim = 0$ at $\pia = 0.95$ in the last two rows, the overall unconditional coverage rate is still above $95\%$ as a result of the overcoverage of $\ciptf = \cin$ given $\phim=1$ and the small weight attached to the case of $\phim=0$ in the mixture.

Importantly, the anticonservative or almost exact unconditional coverage rates  of $\ciptf$ in Table \ref{tb:cr} are due to the way we generate the finite population. Overcoverage by $\ciptf$ is observed for the majority of other finite populations we investigated, and suggests that the anticonservativeness of $\ciptf$ should in general be minor in practice even when the individual treatment effects are heterogeneous under treatment groups of unequal sizes. 
This is no surprise but the result of the two other counteracting forces in play.
In particular, when $\phim = 0$, we construct $\ciptf = \cif$ using $\hsef$, which is asymptotically conservative for the true sampling variance.
This, together with the overconservativeness of $\cin$ when $\phim = 1$, constitutes two sources of conservativeness that can largely offset the undercoverage due to the limit of $\sqrtn (\htf-\tau) \mid \{\phim = 0\}$ having wider central quantile ranges than $\mn(0,\vf)$,  rendering the actual coverage rate of $\ciptf$ usually at least $1-\alpha$.

\subsection{Validity for testing the weak null hypothesis}
We illustrate in this subsection the invalidity of $\htpts$ and $\htpts/\hsepts \ (\fl)$ for testing $\hn$. 
Consider a finite population of $N=100$ units, $i = \ot{N}$, subject to a treatment-control trial of treatment sizes $(N_0, N_1) = (90, 10)$.
We generate the following two sets of finite population to illustrate the invalidity of the unstudentized $\htpts \ (\fl)$ and the studentized $\htpts/\hsepts \ (\fl)$, respectively:
\begine[$\mathcal{P}_1$:]
\item the covariates $x_i$'s are independent ${\rm Uniform}(-1,1)$, and the potential outcomes are $Y_i(1) \sim \mathcal{N}( x_i^3 , 1)$ and $Y_i(0) \sim \mathcal{N}( -x_i^3 , 0.5^2)$;
\item the covariates  $x_i$'s are independent $\mn(0,1)$, and the potential outcomes are $Y_i(1) = \ep_i$ and $Y_i(0) = x_i + \ep_i$, where $\ep_2, \ldots, \ep_N$ are independent $\mn(0,1)$ and $\ep_1$ is chosen to ensure $\sumi \ep_i x_i = 0$.
\ende
We center $Y_i(0)$'s and $Y_i(1)$'s respectively to ensure $\tau= 0$, such that $\hn$ holds. 

Fixing $\{Y_i(0), Y_i(1), x_i\}_{i=1}^N$ in the simulation, we draw a random permutation of $N_{1}$ 1's and $N_{0}$ 0's to obtain the observed outcomes and conduct the {\frt}.
%
%
We repeat the procedure $1000$ times, with the $\pfrt$ approximated by $500$ independent permutations of the assignment vector in each replication. 
Figure \ref{fig::typeoneerror} shows the distributions of $\pfrt$  under $\HN $ for $\htpts$ and $\htpts/\hsepts \ (\fl)$ under $\mathcal{P}_1$ and $\mathcal{P}_2$, respectively.  
Assume significance level $\alpha = 0.05$. The unstudentized $\htpts \ (\fl)$ fail to preserve the correct type I error rates under  $\mathcal{P}_1$.
The studentized $\htpts/\hsepts \ (\fl)$  fail to preserve the correct type I error rates under  $\mathcal{P}_2$.

\begin{figure}[t!]

\centering

 \begin{subfigure}[t]{0.46\textwidth}
 \centering 
\includegraphics[scale=.8]{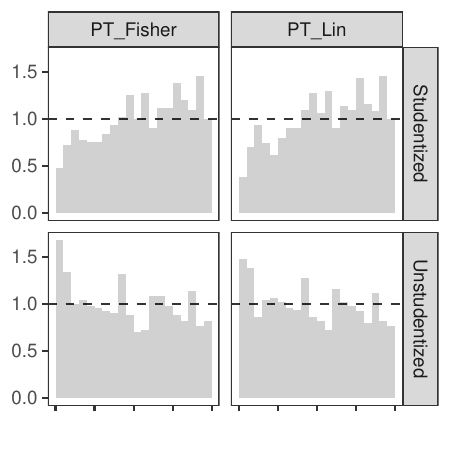}
\caption{Under $\mathcal{P}_1$, the two unstudentized test statistics fail to preserve the correct type I error rates at $\alpha = 0.05$.}
\label{fig:a}
\end{subfigure}
~
\begin{subfigure}[t]{0.46\textwidth}
\centering
\includegraphics[scale=.8]{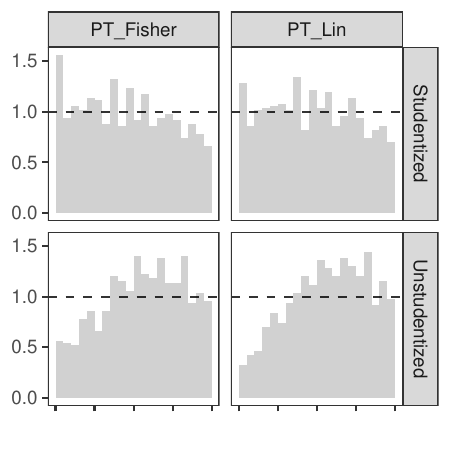}
\caption{Under  $\mathcal{P}_2$, the two studentized test statistics fail to preserve the correct type I error rates at $\alpha = 0.05$.}
\label{fig:b}
\end{subfigure}

\bigskip

\caption{\label{fig::typeoneerror}Empirical histograms of $p_{\frt}$ under $\hn$ with 20 bins in $(0,1)$. }
\end{figure}

\section{Discussion}\label{sec:discussion}
We studied the role of covariate balance in randomized trials and evaluated the properties of the preliminary-test procedure from the randomization-based perspective.
The results discourage the use of preliminary tests for covariate adjustment and corroborate the advantages of regression adjustment by \cite{Lin13}'s specification. 
We list below some possible extensions for future research. 

Importantly, our recommendation of the fully interacted adjustment by \cite{Lin13} is based on large sample considerations. 
In finite samples, the inclusion of the interactions almost doubles the model complexity compared with the additive adjustment, subjecting subsequent inference to  large finite-sample variability. 
We refer interested readers to \cite{zdca} for a systematic discussion.

\subsection{Covariate balance criteria based on statistical tests}
We focused on the Mahalanobis distance-based criterion for measuring covariate balance in defining the preliminary-test procedure without variable selection. 
We conjecture that the results extend to balance criteria based on significance tests \citep{cochran, permutt, hansen2008, zdrep} with minimal modifications. 
That is, we define the balance criterion based on results from one or more significance tests, and adjust for all covariates by $\htf$ or $\htl$ if the realized allocation fails the criterion.

\subsection{Variable selection by preliminary tests based on the outcome model}\label{sec:vs_outcome}
The discussion in Section \ref{sec:vs} focuses on variable selection based on preliminary tests of covariate balance. 
It is a special case of the broader class of statistical procedures that use a test of significance to guide the choice of specification for subsequent analysis \citep{bancroft}. 
An alternative, arguably more popular, strategy is to select the most prognostic covariates based on the outcome model. 
\cite{posthoc} and \cite{permutt} conducted simulation studies to compare these outcome model-based selection criteria with those based on covariate balance, and  concluded that the outcome model-based criteria often lead to superior performances than their covariate balance-based counterparts. 
\cite{raab}, on the other hand, suggested that the outcome model-based criteria may lead to overestimation of the treatment effect. 
\cite{hanzhong} studied  regression adjustment by lasso in high-dimensional settings, which is essentially a variable selection procedure based on the outcome model. 
We leave the theoretical comparison of such approaches with their covariate balance-based counterparts to future research.

\subsection{Covariate adjustment in nonlinear regression models}\label{sec:ext_nonlinear}
We focused on covariate adjustment by linear regression. 
Covariate adjustment in nonlinear regression models can be very different \citep{robinson}, where even omission of a balanced prognostic covariate can lead to bias \citep{christensen, chastang, altman3, raab94}.
\cite{beach} examined covariate selection with preliminary test in proportional hazards regression model using a simulation study, and concluded that adjustment can often decrease the precision under conditions typical in clinical trials. 
\cite{freedman} and \cite{guo} laid down the theoretical basis for randomization-based inference from nonlinear regression. 
We leave extensions to the preliminary-test procedure to future research.

\subsection{Preliminary test for regression specification}\label{sec:ext_specification}
We assumed that the analyzer predetermines to use either \cite{Fisher35}'s additive or \cite{Lin13}'s fully interacted specification for regression adjustment. 
Alternatively, we can use preliminary tests to determine the regression specifications for covariate adjustment.  An immediate example is to use data to determine whether a subset or all of the interaction terms can be omitted from the regression specification.  Alternatively, we can use preliminary testing to decide whether to use linear or non-linear regression specification. 
See also \cite{hanzhong} for an example of using lasso to determine data-dependent specification of the mean model.
Such model selection parallels the variable selection discussed in Section \ref{sec:vs_outcome} as a special case of the broader class of statistical procedures that use a test of significance to guide the choice of specification for subsequent analysis \citep{bancroft},  and can be performed either independently or in combination with preliminary test of covariate balance. 
We leave detailed theory to future research.

\subsection{Preliminary tests in observational studies}
\cite{cochran} studied the implications of covariate balance-based preliminary tests for covariate adjustment in observational studies, and suggested the practical value of preliminary tests when the number of covariates is large.
\cite{raab94} considered two outcome model-based procedures for variable selection in observational studies, and suggested that they are inferior to the always-adjust counterpart in general. 
This is coherent with our recommendation for the always-adjust strategy under randomized trials in low-dimensional settings. 
In a similar spirit, \cite{zach2} used balance tests to assess if a matched data set should be analyzed as a completely randomized or
rerandomized experiment.
We leave the theory to future research.

%
%
%

\bibliographystyle{chicago}
\bibliography{refs_pt}

\newpage 

\setcounter{equation}{0}
\setcounter{section}{0}
\setcounter{figure}{0}
\setcounter{example}{0}
\setcounter{proposition}{0}
\setcounter{corollary}{0}
\setcounter{theorem}{0}
\setcounter{table}{0}
\setcounter{condition}{0}
\setcounter{lemma}{0}
\setcounter{remark}{0}
\setcounter{definition}{0}

\renewcommand {\thedefinition} {S\arabic{definition}}
\renewcommand {\theproposition} {S\arabic{proposition}}
\renewcommand {\theexample} {S\arabic{example}}
\renewcommand {\thefigure} {S\arabic{figure}}
\renewcommand {\thetable} {S\arabic{table}}
\renewcommand {\theequation} {S\arabic{equation}}
\renewcommand {\thelemma} {S\arabic{lemma}}
\renewcommand {\thesection} {S\arabic{section}}
\renewcommand {\thetheorem} {S\arabic{theorem}}
\renewcommand {\thecorollary} {S\arabic{corollary}}
\renewcommand {\thecondition} {S\arabic{condition}}
\renewcommand {\thepage} {S\arabic{page}}
\renewcommand {\theremark} {S\arabic{remark}}

\setcounter{page}{1}

\begin{center}
\huge Supplementary Material
\end{center}

%
 
Section \ref{sec:notation_app} gives the notation, definitions, and regularity conditions omitted from the main text. 

Section \ref{sec:thm_app} gives the rigorous statements of Theorems \ref{thm:coverage} and \ref{thm:coverage_conditional} in the main text, along with some results on the corresponding $t$-statistics as the building blocks.

Section \ref{sec:correction_app} gives the procedures for correcting for the asymptotic coverage rates and type I error rates under the preliminary-test procedure.

Section \ref{sec:lemmas} states the key lemmas.

Sections \ref{sec:proof_dist_htpt}--\ref{sec:frt_app} give the proofs of the results on the asymptotic sampling distributions, coverage rates, and the Fisher randomization test ({\frt}), respectively.
To facilitate connection, we present the results for the unconditional and conditional inferences together, leading with the conditional part as the building blocks.

Assume centered covariates with $\bar x = \meani x_i = 0_J$ throughout to simplify the presentation. 
Let $\mq = [-\qa, \qa]$ denote the central $1-\alpha$ quantile range of the standard normal distribution with $\pr(\ep \in \mq) = 1-\alpha$ for $\ep\sim \mn(0,1)$. Let $|\ci|$ denote the length of a confidence interval $\ci$.  

\section{Notation, definitions, and regularity conditions}\label{sec:notation_app}

\subsection{Peakedness and asymptotic relative efficiency}

Definition \ref{def:peakedness} below gives the definition of peakedness \citep{sherman} that underlies the definition of asymptotic relative efficiency in the main text. 

\begin{definition}\label{def:peakedness}
For two symmetric random vectors $A$ and $B$ in $\mrm $, we say $A$ is {\it more peaked} than $B$ if $\pr(A\in \mc) \geq \pr(B \in \mc)$ for every symmetric convex set $\mc \in \mrm $, denoted by $A \succeq B$. 
\end{definition}

For $m = 1$, a more peaked random variable has smaller central quantile ranges.
For general $m\geq 1$ and $A, B \in \mathbb R^m$ with finite second moments, $A \succeq B$ implies $\cov(A) \leq \cov(B)$. 
For $A$ and $B$ in $\mathbb R^m$ that are both normal with zero means, $A \succeq B$ is equivalent to $\cov(A) \leq \cov(B)$. 
Intuitively, $A$ and $B$ have the same distribution, denoted by $A \sim B$, if and only if $A \succeq B$ and $B \succeq A$.  

For two sequences of random vectors $\{A_N\}_{N=1}^\infty$ and $\{B_N\}_{N=1}^\infty$ with $A_N\rs A$ and $B_N\rs B$ in $\mrm$, write $A_N \succi B_N$ if $A \succeq B$, and write $A_N \asim B_N$ if $A \sim B$. 

Definition \ref{def:eff2} below was introduced by \cite{zdrep} to quantify the asymptotic relative efficiency of estimators in terms of peakedness.
The definition is equivalent to the one by central quantile ranges in the main text when $m = 1$, but also extends to general $m$-dimensional estimators for $m \geq 2$.

\begin{definition}\label{def:eff2}
For two estimators $\hth_1$ and $\hth_2$ that are both consistent for some parameter $\theta \in \mrm$ as the sample size $N$ tends to infinity, we say 
\begine[(i)]
\item $\hth_1$ and $\hth_2$ are {\it asymptotically equally efficient} if $\sqrtn (\hth_1 - \theta) \asim \sqrtn (\hth_2 -\theta)$; 
\item $\hth_1$ is  {\it asymptotically more efficient} than $\hth_2$ if $\sqrtn (\hth_1 - \theta) \succi \sqrtn (\hth_2 - \theta)$. 
\ende 
 For notational simplicity, we will abbreviate  $\sqrtn (\hth_1 - \theta) \asim \sqrtn (\hth_2 -\theta)$ as $\hth_1 \asim \hth_2$, and $\sqrtn (\hth_1 - \theta) \succi \sqrtn (\hth_2 - \theta)$ as $\hth_1 \succi \hth_2$, respectively, when the meaning of $\theta$ is clear from the context. 
\end{definition}

\subsection{Explicit forms of $v_* \ (\nflt)$ and additional notation}
For $z = 0,1$, let $\gq $ be the coefficient vector of $x_i$ from $\lmt(\yiz \sim 1 + x_i)$.
This is a theoretical fit with the $\yiz$'s only partially observable. 
Let $
\gamma_\fisher = e_0\gamma_0 + e_1\gamma_1$. 

For $i = \ot{N}$, let 
\begina
Y_{i, \nm}(z) = \yiz, \quad Y_{i, \fisher}(z) = \yiz   -  x_i^\T\gp, \quad Y_{i, \lin}(z) =  \yiz   -  x_i^\T\bg_z,
\enda with $\meani Y_{i,*}(z) = \by(z)$ for $\nfl$. 
Let 
\begina
\tau_{i,*} = Y_{i, *}(1) - Y_{i,*}(0)\quad (\nfl)
\enda
with $\tau_{i, \nm} = \tau_{i, \fisher} = \tau_i$ and $\meani \tau_{i,*} = \tau$.
Let 
\begina
S^2_{z,*}= (N-1)^{-1}\sumi \{Y_{i,*}(z) - \by(z)\}^2, \quad S_{\tau,*}^2 =  (N-1)^{-1}\sumi (\tau_{i,*} - \tau)^2
\enda denote the finite population variances of $\{Y_{i,*}(z) : i = \ot{N}\}$ and $\{\tau_{i,*}: i = \ot{N}\}$, respectively, with  $
S_{\tau, \nm}^2= S_{\tau,\fisher}^2 = (N-1)^{-1}\sumi (\tau_i - \tau)^2$. 
Then 
\beginy\label{eq:vs}
v_* = e_0^{-1} S_{0,*}^2 +e_1^{-1} S_{1,*}^2  - S_{\tau,*}^2
\endy 
for $\nfl$. 
Let 
\begina
\kappa_* = \frac{v_*}{v_* + S_{\tau,*}^2} \leq 1, \quad \rho_* = \frac{\vl}{\vs} \leq 1
\enda 
for $\nfl$, with $\rho_\lin = 1$. 

Condition \ref{cond:ct} below formalizes the constant treatment effects condition. 

\begin{condition}\label{cond:ct}
The individual treatment effects $\tau_i$'s are constant across $i = \ot{N}$. 
\end{condition}

Condition \ref{cond:ct} ensures 
\begina
v_\fisher = v_\lin, \quad S_{\tau, *}^2 = 0 \quad (\nfl)  
\enda 
such that 
\begina
\kappa_* = 1 \quad(\nfl), \quad \rho_\fisher = 1.
\enda

\subsection{Regularity conditions for asymptotic analysis}

Let $\sxx = (N-1)^{-1}\sumi x_ix_i^\T$ be the finite population  covariance matrix of the centered covariate vectors.
Condition \ref{cond:asym} gives the regularity conditions for finite population asymptotics under complete randomization, embedding the study population in question into a sequence of finite populations with increasing $N$ \citep{DingCLT}.

\begin{condition}\label{cond:asym}  
As $N \to \infty$, for $z =  0,1$, 
(i) $e_z  = N_z /N$ has a limit in $(0,1)$, 
(ii) the first two finite population moments of $\{Y_i(0), Y_i(1), x_i \}_{i=1}^N $  have finite limits; $\sxx$ and its limit are both nonsingular; 
and (iii) there exists a $ c_0 < \infty$ independent of $N$ such that $N^{-1}\sum_{i=1}^N Y^4_i(z) \leq c_0$ and $N^{-1}\sum_{i=1}^N \|x_i\|^4_4 \leq c_0$.
\end{condition}
Condition \ref{cond:asym} ensures that $e_z $, $\gq$, 
$\gp$, $S_x^2$, $v_*$, and $S_{\tau, *}^2$ all have finite limits for $z = 0,1$ and $\nfl$. 
For notational simplicity, we will also use  the same symbols  to denote their respective limiting values when no confusion would arise. 

Lemma \ref{lem:cre} below reviews the asymptotic distributions of $\hts$ and $\htx$ under complete randomization \citep{Neyman23, Lin13, DingCLT, LD20, zdfrt}.

\begin{lemma}\label{lem:cre} 
Assume a completely randomized treatment-control experiment.
Then $v_x = N\cov(\htx) = \ppinv\sxx$. 
Further assume Condition \ref{cond:asym}. 
Then
\begine[(i)]
\item \begina
\rtn   \left(\begin{array}{cc}
\hts  - \tau \\
\htx
\end{array}
\right)
\ \rightsquigarrow \   
\mN \left\{ 
0_{J+1},  
  \left(\begin{array}{cc}
v_*  &  \csi^\T \\
 \csi & \vxi \end{array}\right) \right\} \quad(\nfl), 
 \enda
with  
\begina
\cni =\sxx(e_0^{-1}\gamma_0 + e_1^{-1} \gamma_1),\qquad  \cfi =   \sxx(e_1^{-1}-e_0^{-1})(\gamma_1-\gamma_0 ),\qquad \cli = 0_J
\enda
satisfying $v_* - \vl  = \cs^\T \vx^{-1} \cs \geq 0$. This ensures $\htl \succi \hts \ (\nf)$ with $\htl \asim \hts$ if and only if $c_* = 0_J$. 
\item $N\hse_*^2 - v_* = S_{\tau, *}^2 + \op$ with $S_{\tau, *}^2 \geq 0$ for $\nfl$. 
\ende
\end{lemma}

\subsection{Background and regularity conditions for the FRT}
Definition \ref{def:proper} below formalizes the definition of asymptotic validity of the {\frt} for testing $\hn$ \citep{wd, zdfrt}.

\begin{definition}\label{def:proper}
A test statistic $T$ is {\it asymptotically valid} for testing $\hn$ if under $\hn$,  
$$
\limN  \pr(  p_\frt \leq \alpha   ) \leq \alpha \qquad \text{for all }\alpha\in(0,1)
$$ 
holds for all sequences of $\{Y_i(0), Y_i(1), x_i: i = \ot{N}\}$ and $(e_0, e_1)$ as $N$ tends to infinity.
\end{definition}

Recall that $\mathcal Z$ denotes the set of all $N!$ permutations of the realized assignment vector $Z$. 
Its nonrepetitive elements constitute the set of all possible treatment assignment vectors of treatment sizes $(N_0, N_1)$.
For each $\bm z \in \mz$, let $Y(\bm z)$ denote the potential value of $Y$ when $Z = \bm z$. 
Complete randomization induces a uniform distribution over $  \{ T(\bm z, Y(\bm z), X) : \bm z \in \mz\}$, known as the \emph{sampling distribution} of $T$.
The {\frt}, on the other hand, induces a uniform distribution over $  \{ T(\bm z, Y(Z), X) : \bm z \in \mz\}$ conditioning on $Z$,  known as the {\it randomization distribution} of $T$ \citep{wd, zdfrt, colin}. 

For a two-sided test with $\pfrt$ computed as in \eqref{eq:pfrt}, a statistic $T$ is asymptotically valid for testing $\hn$  if under $\hn$,  the sampling distribution of $|T|$ is stochastically dominated by its randomization distribution for almost all sequences of $Z$. 

Let $w_i(z) =  (S_x^2)^{-1}x_i Y_i(z)$ for $i = \ot{N}$ and $z = 0,1$. 
Condition \ref{cond:asym_frt} below gives the additional regularity conditions for finite population asymptotics of the randomization distributions.

\begin{condition}\label{cond:asym_frt}  
As $N \to \infty$, for $z =  0,1$, 
(i) 
the second moments of $\{  w_i(0),  w_i(1)\}_{i=1}^N$ have finite limits, 
and (ii) there exists a $ c_0 < \infty$ independent of $N$ such that $N^{-1}\sum_{i=1}^N \|w_i(z)\|_4^4 \leq c_0$.
\end{condition}

\subsection{Regularity conditions for the results in the main text}
Propositions \ref{prop:htpt}--\ref{prop:cond} and Theorems \ref{thm:htpt}, \ref{thm:coverage}, \ref{thm:cond}, and \ref{thm:coverage_conditional} assume Condition \ref{cond:asym}. 
 Theorems \ref{thm:frt} and \ref{thm:frt_cond} assume Conditions \ref{cond:asym}--\ref{cond:asym_frt}.

\section{Rigorous statements and additional results}\label{sec:thm_app}

\subsection{Rigorous statements of Theorems \ref{thm:coverage} and \ref{thm:coverage_conditional}}
Recall the definitions of $\kappa_*$ and $\rho_* \ (\nfl)$ from Section \ref{sec:notation_app}. 
Theorem \ref{thm:coverage_conditional_app} below complements Theorem \ref{thm:coverage_conditional} in the main text, and quantifies the conditional coverage rates of $\cis \ (\nfl)$ and $\cipts   \ (\fl)$ given $\phim$.

\begin{theorem}\label{thm:coverage_conditional_app}
\precreapp. 
\begine[(i)]
\item 
Conditioning on $\{\phim = 1\}$, we have $\cipts = \cin \ (\fl)$ with 
\beginy
 \limn \pr(\tau \in \cis \mid \phim = 1) &=&  \pr\left[\kappa_*^{1/2} \big\{ \rho_* ^{1/2}  \epsilon +  \left(1 -  \rho_* \right)^{1/2}    \ml \big\}  \in \mq\right] \nonumber\\
&\geq &  \pr\left(\kappa_*^{1/2}   \epsilon   \in \mq\right) \label{eq:overconservative_cond_0} \\
&\geq& 1-\alpha \quad(\nf), \label{eq:overconservative_cond}\\\nonumber\\
\limn \pr(\tau \in \cil \mid \phim = 1) &=&   \pr\left(\kappa_\lin^{1/2} \ep \in \mq\right) \ \geq\ 1-\alpha, \label{eq:overconservative_cond_2}\\\nonumber\\
\limn \frac{|\cis|}{|\cil|} &\leq& 1 \quad (\nf).\nonumber
\endy 
The  equality in \eqref{eq:overconservative_cond_0} holds if and only if $c_* = 0_J$ for $\nf$. 
The  equality in \eqref{eq:overconservative_cond} holds if and only if Condition \ref{cond:ct} holds. 
The equality in \eqref{eq:overconservative_cond_2} holds if Condition \ref{cond:ct} holds. 
\item 
Conditioning on $\{\phim = 0\}$, we have $\cipts = \cis \ (\fl)$ with 
\beginy
 \limn \pr(\tau \in \cis \mid \phim = 0) &=&  \pr\left[\kappa_*^{1/2} \big\{ \rho_* ^{1/2}  \epsilon +  \left(1 -  \rho_* \right)^{1/2}    \ml' \big\}  \in \mq\right] \nonumber\\
&\leq &  \pr\left(\kappa_*^{1/2}   \epsilon   \in \mq\right) \label{eq:anticonservative_cond_0} \\
&\geq& 1-\alpha \quad(\nf), \label{eq:anticonservative_cond}\\\nonumber\\
\limn \pr(\tau \in \cil \mid \phim = 0) &=&   \pr\left(\kappa_\lin^{1/2} \ep \in \mq\right)  \ \geq\ 1-\alpha, \label{eq:anticonservative_cond_2}\\\nonumber\\
\limn \frac{|\cis|}{|\cil|} &\leq& 1 \quad (\nf). \nonumber
\endy 
The equality in \eqref{eq:anticonservative_cond_0} holds if and only if  $c_* = 0_J$ for $\nf$. 
The equality in \eqref{eq:anticonservative_cond} holds if and only if Condition \ref{cond:ct} holds. 
The equality in \eqref{eq:anticonservative_cond_2} holds if Condition \ref{cond:ct} holds. 
\item 
Further assume Condition \ref{cond:ct}. 
Then 
\begina
 \limn \pr(\tau \in \cin \mid \phim = 1) &=&  \pr\left\{ \rho_\nm ^{1/2}  \epsilon +  \left(1 -  \rho_\nm \right)^{1/2}    \ml  \in \mq\right\} \\
&\geq &   \pr\left( \epsilon   \in \mq\right) = 1-\alpha,  \\
 \limn \pr(\tau \in \cin \mid \phim = 0) &=&  \pr\left\{  \rho_\nm ^{1/2}  \epsilon +  \left(1 -  \rho_\nm \right)^{1/2}    \ml'  \in \mq\right\} \\
&\leq &  \pr\left( \epsilon   \in \mq\right) = 1-\alpha,
\enda
whereas
\begina
\limn \pr(\tau \in \cis \mid \phim = t) \ =\ \pr\left(  \ep \in \mq\right)  \ = \ 1-\alpha 
\enda 
for $ t = 0,1$ and $\fl$. 
\ende
\end{theorem}

Theorem \ref{thm:coverage_app} below complements Theorem \ref{thm:coverage} in the main text, and clarifies the source of over- and anticonservativeness of $\cipts \ (\fl)$ unconditionally. 

\begin{theorem}\label{thm:coverage_app}
\precreapp. 
\begine[(i)]
\item 
 \begina
\limn \pr(\tau \in \ciptl) &=& \pi_a \limn \pr(\tau \in \cin \mid \phim = 1) + (1-\pia) \limn \pr(\tau \in \cil \mid \phim = 0)\\
&\geq& 1-\alpha, 
\enda
where $
\limn \pr(\tau \in \cin \mid \phim = 1) \geq 1-\alpha$ and 
$\limn \pr(\tau \in \cil \mid \phim = 0)   \geq  1-\alpha$. 
\item \begina
\limn \pr(\tau \in \ciptf) &=& \pi_a \limn \pr(\tau \in \cin \mid \phim = 1) + (1-\pia) \limn \pr(\tau \in \cif\mid \phim = 0), 
\enda
where $
\limn \pr(\tau \in \cin \mid \phim = 1) \geq 1-\alpha$ whereas $\limn \pr(\tau \in \cif \mid  \phim = 0)$ can be less than $1-\alpha$. Two exceptions are $e_0 = e_1 = 1/2$ such that 
\begina
\limn \pr(\tau \in \cif \mid  \phim = 0)  \geq 1-\alpha, \quad \limn \pr(\tau \in \ciptf) \geq 1- \alpha, 
\enda
and Condition \ref{cond:ct} such that 
\begina
\limn \pr(\tau \in \cif \mid  \phim = 0)  = 1-\alpha, \quad \limn \pr(\tau \in \ciptf) \geq 1- \alpha.
\enda

\item 
Further assume Condition \ref{cond:ct}. 
Then 
\beginy
\limn \pr(\tau \in \cin \mid \phim = 1) &\geq& 1-\alpha, \label{eq:conservative_ct}\\
\limn \pr(\tau \in \cis\mid \phim = 0) &= &  1-\alpha \quad (\fl)  \nonumber
\endy
such that 
\beginy
\limn \pr(\tau \in \cipts) &=& \pi_a \limn \pr(\tau \in \cin \mid \phim = 1) + (1-\pia) \limn \pr(\tau \in \cis \mid \phim = 0)\nonumber\\
&\geq& 1-\alpha,\nonumber
\endy
The equality holds if and only if the equality in \eqref{eq:conservative_ct} holds with $c_\nm = 0_J$. That is, $\ciptf$ and $\ciptl$ are both overconservative under Condition \ref{cond:ct}, with the overcoverage by $\cin$ when $\phim = 1$ being the source of overconservativeness. 
\ende
\end{theorem}
\subsection{Asymptotic sampling distributions of  the $t$-statistics}
Observe that 
\begina
\{\tau \in \cis\} &=& \left\{\frac{\hts -\tau}{\hses} \in \mq\right\} \quad (\nfl), \\
 \{\tau \in \cipts\} &=& \left\{\frac{\htpts -\tau}{\hsepts} \in \mq\right\} \quad (\fl).
\enda
Proposition \ref{prop:htpt_studentized_cond} below states the asymptotic conditional distributions of the studentized statistics, providing the basis for understanding Theorems \ref{thm:coverage_conditional_app} and \ref{thm:coverage_app}.

\begin{proposition}\label{prop:htpt_studentized_cond}
\precreapp. 
\begine[(i)]
\item Conditioning on $\{\phim = 1\}$, we have 
\begina
\frac{\htpts-\tau}{\hsepts} = \frac{\htn-\tau}{\hsen} \quad (\fl),
\enda
with 
\begina
\left. \frac{\hts-\tau}{\hses} \ \right| \  \{\phim = 1\} \ \ \rs \ \  
\kappa_*^{1/2} \left\{ \rho_* ^{1/2}  \epsilon +  \left(1 -  \rho_* \right)^{1/2}    \ml \right\} 
\ \ \succeq \ \  \ep \quad (\nfl).
\enda
\item Conditioning on $\{\phim = 0\}$, we have 
\begina
\frac{\htpts-\tau}{\hsepts} = \frac{\hts-\tau}{\hses} \quad (\fl),
\enda
with 
\begina
\left. \frac{\hts-\tau}{\hses} \ \right| \  \{\phim = 0\} &\rs& 
\kappa_*^{1/2} \left\{ \rho_* ^{1/2}  \epsilon +  \left(1 -  \rho_* \right)^{1/2}    \ml' \right\} \quad (\nf),\\
\left. \frac{\htl-\tau}{\hsel} \ \right| \  \{\phim = 0\} &\rs& 
\kappa_\lin^{1/2}   \epsilon \ \ \succeq \ \ \epsilon.
\enda
\item 
Further assume Condition \ref{cond:ct}. Then 
\begina
\left. \frac{\htn-\tau}{\hsen} \ \right| \  \{\phim = 1\} &\rs& 
 \rho_\nm ^{1/2}  \epsilon +  \left(1 -  \rho_\nm \right)^{1/2}    \ml, \\
\left. \frac{\htn-\tau}{\hsen} \ \right| \  \{\phim = 0\} &\rs& 
 \rho_\nm ^{1/2}  \epsilon +  \left(1 -  \rho_\nm \right)^{1/2}    \ml', 
\enda
whereas
\begina
 \left. \frac{\hts-\tau}{\hses} \ \right| \  \{\phim = t\} &\rs&  \ep
\enda
for $t = 0,1$ and $\fl$. 
\ende\end{proposition}

Proposition \ref{prop:htpt_studentized_sampling dist} below parallels Proposition \ref{prop:htpt_studentized_cond}, and states the asymptotic unconditional distributions of the studentized statistics.

\begin{proposition}\label{prop:htpt_studentized_sampling dist}
\precreapp. 
As $N\to \infty$, we have 
\begina
\frac{\hts -\tau}{\hses} & \rs & \kappa_*^{1/2} \ep 
\ \  \succeq \ \  \ep\quad(\nfl), \\
\frac{\htptf-\tau}{\hseptf} &\rs& 
\mix\beginp
\kappa_\neyman^{1/2} \left\{ \rho_\nm ^{1/2}  \epsilon +  \left(1 -  \rho_\nm \right)^{1/2}    \ml \right\}&:&\pi_a\\
\kappa_\fisher^{1/2} \left\{ \rho_\fisher  ^{1/2}  \epsilon +  \left(1 -  \rho_\fisher\right)^{1/2}    \ml'  \right\}&:& 1-\pi_a
\endp,\\\\
\frac{\htptl-\tau}{\hseptl} &\rs& 
\mix\beginp
\kappa_\neyman^{1/2} \left\{ \rho_\nm ^{1/2}  \epsilon +  \left(1 -  \rho_\nm \right)^{1/2}    \ml \right\}&:&\pi_a\\
\kappa_\lin^{1/2}  \epsilon  &:&1-\pi_a 
\endp \ \ \succeq \ \ \ep.
\enda
Further assume Condition \ref{cond:ct}. We have 
\begina
\frac{\hts -\tau}{\hses} & \rs & \ep \quad (\nfl), \\
\frac{\htpts-\tau}{\hsepts} &\rs& 
\mix\beginp
 \rho_\nm ^{1/2}  \epsilon +  \left(1 -  \rho_\nm \right)^{1/2}    \ml&:&\pi_a\\
  \epsilon &:& 1-\pi_a
\endp
 \ \ \succeq \ \  \ep\quad (\fl).
\enda
\end{proposition}

\subsection{Randomization distributions of the test statistics}
Recall the definitions of $(S_{0,*}^2, S_{1,*}^2)$ from Section \ref{sec:notation_app} for $\nfl$. Let 
\begina
\tvn  =e_1^{-1} S_{0,\neyman}^2 + e_0^{-1}S_{1,\neyman}^2 + \tau^2, \quad \tvl   =  e_1^{-1} S_{0, \fisher}^2 +e_0^{-1}S_{1, \fisher}^2 + \tau^2
\enda with   
$\tilde \rho_\nm = \tvl   / \tvn  \leq 1$. 
Proposition \ref{prop:htpt_rand_cond} below states the asymptotic randomization distributions of $\htpts$ and $\htpts/\hsepts \ (\fl)$ conditional on $\phim$.
The result implies Theorem \ref{thm:frt_cond} in the main paper. 

Let  $T^\pi$  denote a random variable from the randomization distribution of $T $ conditioning on $Z$. 
\begin{proposition}\label{prop:htpt_rand_cond}
\prefrt. 
\begine[(i)]
\item Conditioning on $\{\phim = 1\}$, we have 
\begina
\htpts^\pi =\htn^\pi, \quad \left(\frac{\htpts}{\hsepts}\right)^\pi = \left(\frac{\htn}{\hsen}\right)^\pi \quad (\fl),
\enda
with 
\begina
\sqrtn \htpts^\pi \mid \{\phim = 1\} &\rs& 
\tvl ^{1/2}  \epsilon +  \left(\tvn -  \tvl \right)^{1/2}    \ml, \\
\left. \left(\frac{\htpts}{\hsepts}\right)^\pi \ \right| \  \{\phim = 1\} &\rs& 
\trn ^{1/2}  \epsilon +  \left(1 -  \trn \right)^{1/2}    \ml 
\enda
 for almost all sequences of $Z$. 
\item Conditioning on $\{\phim = 0\}$, we have 
\begina
\htpts^\pi =\hts^\pi, \quad \left(\frac{\htpts}{\hsepts}\right)^\pi = \left(\frac{\hts}{\hses}\right)^\pi \quad (\fl),
\enda
with 
\begina
\sqrtn \htpts^\pi \mid \{\phim = 0\} \ \rs \ 
\tvl ^{1/2} \epsilon,\quad 
\left. \left(\frac{\htpts}{\hsepts}\right)^\pi \ \right| \  \{\phim = 0\} \ \rs \ 
 \epsilon 
\enda
 for almost all sequences of $Z$. 
\ende
\end{proposition}

Proposition \ref{prop:htpt_studentized_randomization dist} below parallels Proposition \ref{prop:htpt_studentized_sampling dist}, and states the asymptotic unconditional randomization distributions of $\htpts$ and $\htpts / \hsepts \ (\fl)$. 
The result implies Theorem \ref{thm:frt} in the main paper. 



\begin{proposition}\label{prop:htpt_studentized_randomization dist}
\prefrt. 
As $N\to \infty$, we have
\begina
 \sqrtn \htpts ^\pi &\rs& 
\mix\beginp
\tvl  ^{1/2} \epsilon +  (\tvn - \tvl)^{1/2}  \ml&: & \pi_a\\
\tvl  ^{1/2} \epsilon   &: & 1-\pi_a
\endp\\
\left( \frac{\htpts}{\hsepts}\right)^\pi &\rs& 
\mix\beginp
\tilde\rho_\nm^{1/2}  \epsilon +  \left(1-\tilde \rho_\nm \right) ^{1/2}   \ml &:& \pi_a\\
\epsilon&:& 1-\pi_a
\endp 
\enda  
 for almost all sequences of $Z$ for $\fl$. 
\end{proposition}

\section{Corrections for confidence intervals and FRT}\label{sec:correction_app}
\subsection{Preliminary test-specific confidence intervals}
We provide in this subsection the details for constructing preliminary test-specific confidence intervals. 

Recall the asymptotic conditional distributions of $\htptl$ from Proposition \ref{prop:cond}:
\beginy\label{eq:htptl_cond}
\sqrtn (\htptl - \tau) \mid \{M <a\} &=& \sqrtn (\htn - \tau) \mid \{M <a\} 
 \ \ \rs \ \    v_\lin  ^{1/2} \epsilon +  (  v_\nm - v_\lin)^{1/2}  \ml,  \nonumber\\
\sqrtn (\htptl - \tau) \mid \{M \geq a\} &=&\sqrtn (\htl - \tau) \mid \{M \geq a\} \ \ 
 \rs \ \  v_\lin^{1/2} \epsilon.
\endy
Rather than using $\ciptl = \cin$ as the confidence interval when $M<a$, 
we can instead estimate the unknown $\vn$ and $\vl$ in \eqref{eq:htptl_cond}
 using $\hat v_\nm = N\hsen^2$ and $\hat v_\lin = N\hsel^2$, 
respectively,
and construct the confidence interval based on the central $1-\alpha$ quantile range of the plug-in distribution
$
\hat v_\lin  ^{1/2} \epsilon +  (\hat v_\nm -\hat v_\lin)^{1/2}  \ml$
\citep{LD2018}. 
This modification mitigates the overconservativeness of $\ciptl$ both unconditionally and when conditional on $\phim = 1$.
The resulting confidence interval is nevertheless still wider than $\cil$ asymptotically.

We can similarly fix the possible anticonservativeness of $\ciptf$. In particular, let $\hat v_\fisher = N\hsef^2$ be an asymptotically conservative estimator of $v_\fisher$. We can construct the confidence interval based on the central $1-\alpha$ quantile range of 
$
\hat v_\lin  ^{1/2} \epsilon +  (\hat v_\nm -\hat v_\lin)^{1/2}  \ml$ as the plug-in estimate of the asymptotic conditional distribution of $\sqrtn (\htptf-\tau)$ when $M < a$ by Proposition \ref{prop:cond}, 
and based on that of
$
\hat v_\lin  ^{1/2} \epsilon +  (\hat v_\fisher -\hat v_\lin)^{1/2}  \ml'$ as the plug-in estimate of the asymptotic conditional distribution of $\sqrtn (\htptf-\tau)$ when $M \geq a$.

\subsection{Prepivoting for valid FRTs of the weak null hypothesis} 
In brief, prepivoting first transforms a test statistic $T = T(Z, Y, X)$ by its estimated cumulative distribution function $\hat F(\cdot)$ constructed from the data to form a new test statistic $T'(Z, Y, X)= \hat F(T)$,
and then uses $T'(Z, Y, X)$ as the test statistic to compute the $p$-value from a one-sided {\frt} as
\beginy\label{eq:p_prepivot}
\pfrt' = |\mz|^{-1}\sum_{\bm z \in \mz} \mi \big\{ T'(\bm z, Y, X) \geq T'(Z, Y, X) \big\}.
\endy 
It recovers studentization when studentization alone is sufficient, and also accommodates other situations where studentization fails.
The resulting {\frt} simultaneously delivers finite-sample exact inference for $\hf$ and asymptotically valid inference for $\hn$.

In particular, given $(\htpts, \hsepts)$ from the preliminary-test procedure for $\fl$, we can take $T = \sqrtn |\htpts|$ as the initial test statistic, and form the new statistic as
\begina
T' = \hat F_* \left( \sqrtn|\htpts| \right),
\enda
where $\hat F_*$  is the estimated cumulative distribution function of $\sqrtn |\htpts|$ under $\hn$.
Recall from Proposition \ref{prop:htpt} that under $\hn$, $\sqrt N \htpts$ converges in distribution to a mixture of $\vl  ^{1/2} \epsilon +  (\vn - \vl)^{1/2}  \ml $ and
$\vl  ^{1/2} \epsilon +  (\vs - \vl)^{1/2}   \ml' $ with weights $(\pi_a, 1-\pi_a)$.
We can hence estimate $\hat F_*$ as the cumulative distribution function of the mixture of $|\hvl  ^{1/2} \epsilon +  (\hvn - \hvl)^{1/2}  \ml|$ and
$|\hvl  ^{1/2} \epsilon +  (\hvs - \hvl)^{1/2}   \ml'|$ with weights $(\pi_a, 1-\pi_a)$.

Alternatively, we can also start with the studentized $T = |\htpts| / \hsepts$, and define 
\begina
T' = \hat F_* \left( \frac{ |\htpts|}{\hsepts} \right),
\enda
where $\hat F_*$ is the estimated cumulative distribution function of $|\htpts| / \hsepts$ under $\hn$. 
Proposition \ref{prop:htpt_studentized_sampling dist} ensures that under $\hn$, 
\begina
\frac{|\htpts|}{\hsepts} &\rs& 
\mix\beginp
\kappa_\neyman^{1/2} \left| \rho_\nm ^{1/2}  \epsilon +  \left(1 -  \rho_\nm \right)^{1/2}    \ml \right|&:&\pi_a\\
\kappa_*^{1/2} \left| \rho_* ^{1/2}  \epsilon +  \left(1 -  \rho_*\right)^{1/2}    \ml'  \right|&:& 1-\pi_a
\endp.
\enda
The corresponding $\hat F_*$ can be constructed by replacing the unknown  $(\kappa_\nm, \kappa_*)$ and $(\rho_\nm, \rho_* )$ with 
$
\hat \kappa_\nm = \hat \kappa_* = 1$, $\hat\rho_\nm =  \hvl/\hvn$, and $\hat\rho_* =  \hvl/\hvs$, 
respectively.

Let $U\sim \textup{Uniform}(0,1)$ denote a uniform random variable on $[0,1]$. 
Let $\tilde U$ denote a random variable that takes values in $[0,1]$ and satisfies $\pr(\tilde U \leq t) \geq t$ for $t \in [0,1]$. 
We have $\tilde U$ is stochastically dominated by $U$. 
Theorem \ref{thm:prepivot} below follow from \citet[Theorem 1 and Corollary 1]{colin}, and ensure the asymptotic validity of the prepivoted test statistics for testing $\hn$.

\begin{theorem}\label{thm:prepivot}
\prefrt. 
As $N\to \infty$, for $T' = \hat F(T)$ as the prepivoted test statistic constructed from $T \in\{ \sqrtn |\htpts|, |\htpts|/\hsepts:\fl\}$, we have 
\begine[(i)]
\item $T'  \rs  \tilde U$; 
\item $(T')^\pi \rs U$  for almost all sequences of $Z$. 
\ende
This ensures that the $\pfrt'$ in \eqref{eq:p_prepivot} is finite-sample exact for testing $\hf$ and asymptotically valid for testing $\hn$. 
\end{theorem}

The same discussion extends to the conditional {\frt} with minimal modification, such that prepivoting ensures the asymptotic validity of 
\begina
\pfrt'(\phim) = |\mz(\phim) |^{-1}\sum_{\bm z \in \mz(\phim)} \mi \big\{ T'(\bm z, Y, X) \geq T'(Z, Y, X) \big\}.
\enda 
 for testing $\hn$ conditioning on $\phim$.

\section{Lemmas}\label{sec:lemmas}
\subsection{Conditional asymptotics}

Let $\phi(B, C)$ be a binary covariate balance indicator function, where $\phi(\cdot, \cdot)$ is a binary indicator function and $(B, C)$ are two statistics computed from the data.  
Lemma \ref{lem:weak_convergence} below is a generalization of \citet[][Proposition A1]{LD2018},  and gives the asymptotic joint distribution of arbitrary random elements conditioning on $\phi(B, C)$.

\begin{condition}\label{cond:A1}
The binary indicator function $\phi(\cdot , \cdot )$ satisfies:
(i) $\phi(\cdot , \cdot )$ is almost surely continuous;
(ii) for $u \sim \mN(0_J, \Sigma )$,  we have $\pr\{\phi(u,   \Sigma ) = 1\}\in (0,1)$ for all $  \Sigma > 0$, and $\cov\{u \mid  \phi (u,   \Sigma ) = t\}$ is a continuous function of
$ \Sigma$ for $t = 0,1$. 
\end{condition}

\begin{lemma}\label{lem:weak_convergence}
Assume that $\phi$ is a binary indicator function that satisfies Condition \ref{cond:A1}. 
For a sequence of random elements $(A_N, B_N, C_N)_{N=1}^\infty$ that satisfies $(A_N, B_N, C_N) \rightsquigarrow (A, B  ,C )$ as $N\to \infty$, we have
\begina
(A_N, B_N) 
 \mid   \{\phi(B_N , C_N) = t \} \ \rs \ 
(
A , B )
 \mid    \{\phi(B , C) = t \}
\enda
for $t = 0,1$
in the sense that, for any continuity set $\mathcal{S}$ of $(A , B ) \mid \{\phi(B , C) = t\}$, 
\begina
\pr\{ (A_N, B_N) \in \mathcal{S} \mid \phi(B_N , C_N) = t\} = 
\pr\{ (A , B ) \in \mathcal{S} \mid \phi(B , C) = t \} + o(1).
\enda
\end{lemma}

%
%

\begin{lemma}\label{lem:cond}
\precreapp. 
\begine[(i)]
\item Conditioning on $\{\phim = 1\}$,  we have 
\begina
\rtn  ( \hts - \tau)  \mid \left\{\phim = 1 \right\}&\rs&    \vl  ^{1/2} \epsilon +  (v_*-\vl)^{1/2} \ml \quad(\nfl).
\enda
\item Conditioning on $\{\phim = 0\}$,  we have  
\begina
\rtn  ( \hts - \tau)  \mid \left\{\phim = 0 \right\}&\rs&   \vl  ^{1/2} \epsilon +  (v_*-\vl)^{1/2} \ml' \quad(\nfl). 
\enda
\ende
\end{lemma}

\begin{proof}[Proof of Lemma \ref{lem:cond}]
The results on $\sqrtn (\hts-\tau) \mid \{\phim = 1\}$ 
follow from \cite{zdfrt} based on Lemma \ref{lem:weak_convergence}.
The proof for $\sqrtn (\hts-\tau) \mid \{\phim = 0\} $ is almost identical and hence omitted. 
\end{proof}

Lemma \ref{lem:se} follows from \citet[Lemma A16]{LD2018}, and ensures that $N\hses^2$ converges in probability to $v_* + S_{\tau, *}^2$ given $\phim$. 

\begin{lemma}\label{lem:se}
\precreapp. For all $\epsilon > 0$ and $\nfl$, we have  
\begina
&&\pr\left\{ \left. \left| N\hses^2 - ( \vs + S_{\tau, *}^2)\right| > \epsilon \ \right| \  \phim = 0 \right\} = \op, \\
&&\pr\left\{ \left. \left| N\hses^2 - ( \vs + S_{\tau, *}^2)\right| > \epsilon\ \right| \  \phim = 1 \right\} = \op. 
\enda
\end{lemma}

\subsection{Peakedness}
Lemma \ref{lem:gci} below states the celebrated Gaussian correlation inequality, with the recent breakthrough proof due to \cite{royen}; see also \cite{inbook}. 
\cite{zdrep} used it in developing the theory of rerandomization based on $p$-values.

\begin{lemma}[Gaussian correlation inequality]\label{lem:gci}
Let $\mu$ be an $m$-dimensional Gaussian probability measure on $\mrm$, that is, $\mu$ is a multivariate normal distribution, centered at the origin. Then $\mu(\mc_1 \cap \mc_2) \geq \mu(\mc_1) \mu(\mc_2)$ for all convex sets $\mc_1, \mc_2 \subset \mrm$ that are symmetric about the origin.  
\end{lemma}

%


Lemma \ref{lem:peak_sum} below reviews two classical results in probability for comparing peakedness between random vectors.
\cite{AOS} and \cite{zdrep} used them before. 
The proofs follow from \citet[][Lemma 7.2 and Theorem 7.5]{dharmadhikari1988}.

\begin{lemma}\label{lem:peak_sum}
\begine[(i)]
\item 	If two $m\times 1$ symmetric random vectors $A$ and $B$ satisfy $A \succeq B$, then $CA \succeq CB$ for any matrix $C$ with compatible dimensions.
\item 
	Let $A$, $B_1$, and $B_2$ be three independent $m\times 1$   symmetric random vectors. If $A$ is  normal and $B_1 \succeq B_2$, then $A+B_1 \succeq A+B_2$. 
\ende
\end{lemma}

Lemma \ref{lem:tdl} below clarifies the relative peakedness of $\ep$, $\ml$, and $\ml'$. 
\begin{lemma}\label{lem:tdl}
$\ml \succeq \epsilon \succeq \ml' $. 
\end{lemma}

\begin{proof}[Proof of Lemma \ref{lem:tdl}]
The result on $\ml \succeq \ep$ follows from \cite{LD2018}. We verify below $\epsilon \succeq \ml'$. 

To begin with, Lemma \ref{lem:gci} ensures that 
\begina
\pr\left(D \in \mc, \ D^\T D <a \right) \geq \pr(D \in \mc) \cdot \pr(D^\T D <a ) 
\enda
for all convex sets $\mc \subset \mathbb R^m$ that are symmetric about the origin. 
Let $A \sim D  \mid \{D^\T D \geq a\}$ be a shorthand for the truncated normal random vector. We have
\begina
\pr(A \in \mc) = \pr\left(D \in \mc \mid D^\T D \geq a \right) &=& 
  \frac{\pr\left(D \in \mc, \ D^\T D \geq a \right)}{\pr\left(  D^\T D \geq a \right)}\\
  &=& 
    \frac{\pr\left(D \in \mc \right) - \pr\left(D \in \mc, \ D^\T D <a \right)}{ 1-\pr\left(  D^\T D <a \right)}\\
    &\leq & 
    \frac{\pr\left(D \in \mc \right) -  \pr\left(D \in \mc\right) \cdot\pr\left( D^\T D <a  \right)}{ 1-\pr\left(  D^\T D <a \right)}\\
    &=&\pr(D \in \mc ),
\enda
such that $D \succeq A$. 
Let $C = (1, 0_{J-1})^\T$ with $\ml'  = C A$ and $D_1 = C D  \sim \mn(0,1) \sim \ep$. 
It then follows from Lemma \ref{lem:peak_sum} that $\ep \sim D_1 = CD \succeq CA = \ml' $. 

\end{proof}

Lemma \ref{lem:peak_mix} below states the relative peakedness of two mixture distributions. The proof follows directly from the definition of peakedness and is hence omitted.

\begin{lemma}\label{lem:peak_mix}
	Let $A$, $B$, $A'$, and $B'$ be four $m\times 1$   symmetric random vectors with $A \succeq A'$ and $B \succeq B'$. We have
	\begina
	\mix\beginp
	A &:& \pi\\
	B &:& 1-\pi
	\endp \ \succeq \ \mix\beginp
	A' &:& \pi\\
	B' &:& 1-\pi
	\endp 
\enda 
for all $\pi \in [0,1]$. 
\end{lemma}

%

\section{Proof of the results on asymptotic distributions and coverage rates}\label{sec:proof_dist_htpt}
\subsection{Point estimators}

\begin{proof}[Proof of Proposition \ref{prop:cond}]
The results follow from Lemma \ref{lem:cond}. 
\end{proof}

\begin{proof}[Proof of Theorem \ref{thm:cond}]  
Recall that 
\begina
\sqrt N(\htl-\tau)   \mid \{\phim = t\}  &\rs& \vl  ^{1/2} \epsilon \quad (t = 0,1)
\enda
by Lemma \ref{lem:cond} and 
\begina
 \sqrt N(\htpts-\tau)   \mid \{\phim = 1\} 
&\rs&    \vl  ^{1/2} \epsilon +   (\vn-\vl)^{1/2}  \ml \quad (\fl),\\
 \sqrt N(\htptf-\tau)   \mid \{\phim = 0\} 
&\rs&    \vl  ^{1/2} \epsilon +   (\vf-\vl)^{1/2}  \ml' 
\enda
by Proposition \ref{prop:cond}.
Lemma \ref{lem:peak_sum} ensures that 
\beginy\label{eq:bridge_2}
\vl  ^{1/2} \epsilon \ \succeq \ \vl  ^{1/2} \epsilon +  (\vn-\vl)^{1/2} \ml , \quad 
\vl  ^{1/2} \epsilon \ \succeq \ \vl  ^{1/2} \epsilon +  (\vf-\vl)^{1/2} \ml'
\endy
 such that $\htl \succi \htpts \ (\fl)$ given $\phim = 1$ and $\htl = \htptl \succi \htptf$ given $\phim = 0$. 
\end{proof}

\begin{proof}[Proof of Proposition \ref{prop:htpt}]
First, 
\beginy\label{eq:cdf}
\pr\left\{\sqrt N(\htpts-\tau) \leq t\right\} 
&=& \pr\left\{\sqrt N(\htpts-\tau) \leq t \mid M < a\right\} \cdot \pr(M < a)\nonumber\\
&& + \pr\left\{\sqrt N(\htpts-\tau)\leq t \mid M \geq a\right\} \cdot \pr(M \geq a). 
\endy
Next, it follows from Lemma \ref{lem:cre} that 
\beginy\label{eq:pia}
\pr(M<a ) = \pi_a + o(1).
\endy 
By \eqref{eq:cdf}--\eqref{eq:pia} and Proposition \ref{prop:cond},  we have
\begina
\limn \pr\left\{\sqrt N(\htpts-\tau) \leq t\right\} 
&=& \pi_a  \cdot   \pr\left\{  \vl  ^{1/2}  \epsilon +  (\vn - \vl)^{1/2}  \ml  \leq t \right\}\\
  && +(1-\pi_a ) \cdot \pr\left\{ \vl  ^{1/2}  \epsilon +  (\vs - \vl)^{1/2} \ml'  \leq t\right\} .
  \enda
\end{proof}

\begin{proof}[Proof of Theorem \ref{thm:htpt}]  
First, 
recall that $
\sqrt N(\htl-\tau) \rs  \vl  ^{1/2} \epsilon$
by Lemma \ref{lem:cre} and 
\begina
 \sqrt N(\htpts-\tau)    
&\rs&    
\mix\beginp
\vl  ^{1/2} \epsilon +   (\vn-\vl)^{1/2}  \ml &:& \pia\\
\vl  ^{1/2} \epsilon +   (\vs-\vl)^{1/2}  \ml' &:& 1-\pia
\endp\quad (\fl)
\enda
by Proposition \ref{prop:htpt}.
 Lemma \ref{lem:peak_mix}  and \eqref{eq:bridge_2} together ensure 
\begina
\vl^{1/2} \ep \  \succeq \ \mix\beginp
\vl  ^{1/2} \epsilon +   (\vn-\vl)^{1/2}  \ml &:& \pia\\
\vl  ^{1/2} \epsilon +   (\vs-\vl)^{1/2}  \ml' &:& 1-\pia
\endp\quad (\fl),
\enda
such that $\htl \succi \htpts$ for $\fl$.

Next, recall that $\ml \sim D_1 \mid \{D^\T D <a\}$, where $D_1 \sim \mn(0,1)$ is the first element of $D \sim \mn(0_J, I_J)$. 
Assume without loss of generality that $D$ is independent $\epsilon$, such that 
$
\vl ^{1/2}  \ep + (\vn-\vl)^{1/2}  D_1 \sim \vn ^{1/2} \ep
$. 
It then follows from Lemmas \ref{lem:peak_sum} and \ref{lem:tdl} that $\ml \succeq D_1$ and hence 
\beginy\label{eq:htptl}
\vl  ^{1/2} \epsilon +  (\vn-\vl)^{1/2}  \ml  \ \succeq \ \vl  ^{1/2} \epsilon +  (\vn-\vl)^{1/2}  D_1 \ \sim \ \vn ^{1/2}  \epsilon.
\endy 
This, together with $\vl ^{1/2}\ep \succeq \vn ^{1/2} \ep$, ensures 
\begina
\mix\beginp
\vl  ^{1/2} \epsilon +   (\vn-\vl)^{1/2}  \ml &:& \pia\\
\vl  ^{1/2} \epsilon   &:& 1-\pia
\endp \ \succeq \ \vn^{1/2}\epsilon
\enda
by Lemma \ref{lem:peak_mix}, 
and hence $\htptl \succi \htn$. 
\end{proof}

\subsection{$t$-statistics}

\begin{proof}[Proof of Proposition \ref{prop:htpt_studentized_cond}]
The asymptotic distributions follow from Proposition \ref{prop:cond} and Lemma \ref{lem:se} by Slutsky's theorem. 
That $
(\hts-\tau) / \hses \mid \{\phim = 1\} \succi \ep \ (\nfl)$ follow from 
\beginy
\kappa_*^{1/2} \left\{ \rho_* ^{1/2}  \epsilon +  \left(1 -  \rho_* \right)^{1/2}    \ml \right\} \quad \succeq \quad  
 \rho_* ^{1/2}  \epsilon +  \left(1 -  \rho_* \right)^{1/2}    \ml \quad \succeq \quad \epsilon \quad(\nfl)\label{eq:peak}
\endy
by Lemmas  \ref{lem:peak_sum} and \ref{lem:tdl}. 
\end{proof}

\begin{proof}[Proof of Proposition \ref{prop:htpt_studentized_sampling dist}]
The asymptotic distributions follow from Proposition \ref{prop:htpt} and Lemma \ref{lem:se} by Slutsky's theorem. 
By Lemma \ref{lem:peak_mix}, that $
(\htptl-\tau) / \hseptl \succi \ep$ follows from \eqref{eq:peak}
and $
\kappa_\lin^{1/2}  \epsilon  \succeq   \epsilon$
by Lemma  \ref{lem:peak_sum}. 
\end{proof}

\subsection{Coverage rates}\label{sec:cr_app}

\begin{proof}[Proof of Theorem \ref{thm:coverage_conditional} and Theorem \ref{thm:coverage_conditional_app}]
Observe that
\begina
\limn \pr(\tau \in \cis \mid \phim = t) 
= \limn \pr\left( \left. \frac{\hts - \tau}{\hses} \in \mq \ \right| \   \phim = t\right)
\enda
 for $t= 0,1$ and $\nfl$. 
The results follow from Proposition \ref{prop:htpt_studentized_cond}.

\end{proof}

\begin{proof}[Proof of Theorem \ref{thm:coverage} and Theorem \ref{thm:coverage_app}]
The result follows from 
\begina
\pr(\tau \in \cipts) = \pr(\phim = 1) \cdot \pr(\tau \in \cin \mid \phim = 1) + \pr(\phim = 0)\cdot \pr(\tau \in \cis \mid \phim = 0) \quad (\fl),
\enda
$\limn \pr(\phim = 1) = \pia$, and Theorem \ref{thm:coverage_conditional_app}. 
\end{proof}

\section{Proof of the results on the FRT}\label{sec:frt_app}
\subsection{Conditional FRT}
\begin{proof}[Proof of Proposition \ref{prop:htpt_rand_cond}]
By \cite{zdfrt},  all results for the sampling distributions under Condition \ref{cond:asym} extend to the randomization distributions under Conditions \ref{cond:asym} and \ref{cond:asym_frt} after we replace the true finite population with the pseudo finite population $\{Y_i'(0), Y_i'(1): i = \ot{N}\}$ with $Y_i'(0) = Y_i'(1) = Y_i$, $\tau' = 0$, 
and 
\begina
\bar Y'(z) = e_0 \by(0) + e_1 \by(1) + o(1), \quad (S^2_z)' = p_0S_0^2 + p_1S_1^2 + e_0e_1\tau^2 +o(1) \quad (z= 0, 1) 
\enda
for almost all sequences of $Z$. 
Direct algebra shows that the counterparts of $v_* \ (\nfl)$ based on the pseudo finite population equal 
$\tvn  =e_1^{-1} S_{0, \nm}^2 + e_0^{-1}S_{1,\nm}^2 + \tau^2$ and $\tvf  = \tvl   =  e_1^{-1} S_{0, \fisher}^2 +e_0^{-1}S_{1, \fisher}^2 + \tau^2$, 
and the counterparts of $\kappa_* \ (\nfl)$ and $\rho_* \ (\fl)$ all equal 1. 
The results follow from Proposition \ref{prop:cond} and Proposition \ref{prop:htpt_studentized_cond}. 
\end{proof}

\begin{proof}[Proof of Theorem \ref{thm:frt_cond}]
We verify below the result for  $\htpts/\hsepts \ (\fl)$. 
The proof for $\htpts \ (\fl)$ are similar and omitted. 

Juxtapose  Proposition \ref{prop:htpt_studentized_cond} with Proposition \ref{prop:htpt_rand_cond}. Let 
$
\mt \sim  
\kappa_\neyman^{1/2}  \{ \rho_\nm ^{1/2}  \epsilon +  \left(1 -  \rho_\nm \right)^{1/2}    \ml \} 
$
be a random variable following the asymptotic sampling distribution of $\htpts/\hsepts$ given $\phim = 1$.
Let
$
\mt '   \sim 
\tilde\rho_\nm^{1/2}  \epsilon +  \left(1-\tilde \rho_\nm \right) ^{1/2}   \ml  
$
be a random variable following the asymptotic randomization distribution of $\htpts/\hsepts$ given $\phim =1$.
Conditioning on $\phim = 1$, the studentized $\htpts/\hsepts$ would be able to preserve the correct type I error rates under $\hn$ if $|\mt'|$ 
stochastically dominates $|\mt|$.
The uncertainty in the relative magnitude of $\rho_\nm$ and $\trn$  nevertheless suggests that this is not guaranteed for all sequences of $\{Y_i (0), Y_i(1), x_i\}_{i=1}^N$.
This implies that 
$\htpts/\hsepts$ is not asymptotically valid for testing $\hn$ given $\phim =1$. 

Similarly for the results under $\phim = 0$. 
In particular, the asymptotic validity of $\htptl/\hseptl$ given $\phim =0$ follows from 
\begina
\left. \frac{\htptl-\tau}{\hseptl} \ \right| \  \{\phim = 0\}  \rs  
\kappa_\lin^{1/2}   \epsilon, \quad \left. \left(\frac{\htptl-\tau}{\hseptl}\right)^\pi \ \right| \  \{\phim = 0\}  \rs 
\epsilon,
\enda
where $|\ep|$ stochastically dominates $|\kappa_\lin^{1/2}\ep|$. 
The asymptotic invalidity of $\htptf/\hseptf$ given $\phim = 0$ follows from 
\begina
\left. \frac{\htptf-\tau}{\hseptf} \ \right| \  \{\phim = 0\}  &\rs&  
\kappa_\fisher^{1/2} \left\{ \rho_\fisher ^{1/2}  \epsilon +  \left(1 -  \rho_\fisher \right)^{1/2}    \ml' \right\},\\
 \left. \left(\frac{\htptf-\tau}{\hseptf}\right)^\pi \ \right| \  \{\phim = 0\}  &\rs& 
\epsilon 
\enda
where $|\ep|$ does not necessarily stochastically dominates $|\kappa_\fisher^{1/2} \{ \rho_\fisher ^{1/2}  \epsilon +  (1 -  \rho_\fisher )^{1/2}    \ml' \}|$. 
\end{proof}

\subsection{Unconditional FRT}

\begin{proof}[Proof of Proposition \ref{prop:htpt_studentized_randomization dist}] 
The result follows from Proposition \ref{prop:htpt} and Proposition \ref{prop:htpt_studentized_sampling dist} by the same reasoning as the proof of Proposition \ref{prop:htpt_rand_cond}. 
\end{proof}

\begin{proof}[Proof of Theorem \ref{thm:frt}]
The result follows from Proposition \ref{prop:htpt_studentized_randomization dist} by the same reasoning as the proof of Theorem \ref{thm:frt_cond}. 
\end{proof}

\end{document}